\documentclass[%
 reprint,
 amsmath,amssymb,amsthm,amsfonts,
 aps,
 pra,
]{revtex4-2}

\usepackage{graphicx}
\graphicspath{ {./figures/} }

\usepackage{dcolumn}
\usepackage{amsthm}
\usepackage{graphicx}
\usepackage{bm}
\usepackage{braket}
\usepackage{mathtools}
\usepackage{soul}
\usepackage{comment}
\usepackage{xcolor}
\usepackage{hyperref}
\hypersetup{
    colorlinks=true,
    linkcolor=blue,
    urlcolor=blue,
    citecolor=blue,
}

\newtheorem{theorem}{Theorem}
\newtheorem{corollary}{Corollary}[theorem]
\newtheorem{lemma}{Lemma}[section]
\theoremstyle{definition}

\newcommand{\pf}[1]{\mathcal{P\!\!F}_{\!{#1}}}

\begin{document}

\preprint{APS/123-QED}

\title{Taming coherent noise with teleportation}

\author{Kathleen (Katie) Chang*}
\author{Qile Su*}
\author{Shruti Puri}
\affiliation{Department of Physics and Applied Physics, Yale University, New Haven, CT 06511, USA}

\date{\today}

\begin{abstract}
Compared to the more widely studied Pauli errors, coherent errors present several new challenges in quantum computing and quantum error correction (QEC). For example, coherent errors may interfere constructively over a long circuit and significantly increase the overall failure rate compared to Pauli noise. Additionally, there is so far no analytical proof for a topological code threshold under coherent errors. Moreover, it is hard to even numerically estimate the performance of QEC under coherent errors as their effect in a Clifford circuit cannot be efficiently classically simulated. In this work, we demonstrate that teleportation effectively tailors coherent errors into Pauli errors, for which analytical and numerical results are abundant. We first show that repeated teleportation of a single qubit decoheres errors, and the average infidelity grows at worst linearly with the number of teleportations, similar to Pauli errors. We then analyze a physically motivated pure $Z$-coherent error model for teleported CSS codes in which over-rotation errors accompany every gate, and find that such an error model is equivalent to a Pauli error model. Our result implies that the performance of a CSS code implemented via teleportation-based error correction or measurement-based error correction with such coherent noise can be efficiently simulated on a classical computer and has an analytically provable threshold. The intrinsic noise-tailoring property of teleportation may ultimately remove the need for randomized compiling in teleportation-based quantum computing schemes.
\end{abstract}

\maketitle


\section{Introduction}

Theoretical results on the performance of a quantum error-correcting protocol rely on assumptions about noise. For 2D topological codes~\cite{bravyi1998quantum, DennisTopoQuantumMemory, kitaev_fault-tolerant_2003}, which are widely used in practice because of their simple geometry, seminal results known as quantum fault-tolerance threshold theorems guarantee near-perfect implementation of a quantum algorithm provided that each underlying physical component fails probabilistically with probability less than a finite threshold value~\cite{DennisTopoQuantumMemory, stace_thresholds_2009, fowler_proof_2012-1}. Such probabilistic noise is called incoherent noise. The most widely studied class of incoherent noise is the Pauli noise in which a noisy quantum operation is modeled by probabilistic application of a Pauli operator after or before a perfect operation. Theoretical interest in this class of noise owes to the fact that the performance of an error correcting code under such noise can be efficiently simulated due to the Gottesman-Knill theorem~\cite{gottesman1997stabilizer, gottesman_heisenberg_1998-1}. 

However, real systems often suffer from noise that is coherent, which can result, for example, from small unitary rotations on qubits due to various experimental imperfections. Some of these imperfections include pulse amplitude drifts~\footnote{Such drifts necessitate frequent recalibrations of control parameters~\cite{arute_quantum_2019} or the use of pulse sequences that are insensitive to amplitude errors~\cite{bluvstein_quantum_2022}.}, frequency fluctuations~\cite{burnett_decoherence_2019}, control-field inhomogeneity across qubits~\cite{bluvstein_quantum_2022}, low-frequency $1/f$ noise~\cite{paladino_mathbsf1mathbsfitf_2014}, crosstalk \cite{huang_microwave_2021, sarovar_detecting_2020, zhao_quantum_2022}, and stray two-qubit interactions~\cite{wei_hamiltonian_2022}. Meticulous hardware engineering, careful calibration, and error-canceling control techniques such as composite pulses~\cite{vandersypen_nmr_2005} and dynamical decoupling \cite{ezzell_dynamical_2023} are routinely performed to mitigate coherent errors. However, residual coherent errors can remain despite these efforts. For example, coherent errors have been found to build up in a recent experiment on the neutral-atom array platform~\cite{bluvstein_logical_2024}. Coherent errors due to stray interactions have also been found to significantly affect error correction experiments on the superconducting platform~\cite{acharya_suppressing_2023, acharya_quantum_2025}. It is therefore important to understand the performance of quantum error correction protocols under coherent errors.

The performance of a quantum error correction code under coherent errors is partially related to its performance under Pauli errors via a theorem known as the discretization of errors~\cite{nielsen_quantum_2010}, which states that linear combinations of correctable errors are also correctable. For stabilizer codes, the intuition behind this theorem is that measuring the stabilizers probabilistically collapses the linear combination into a correctable Pauli error, effectively decohering the coherent error. However, if single-qubit coherent errors occur on all physical qubits, for example, then the tensor product of these errors is a linear combination of both low-weight correctable Pauli errors and high-weight uncorrectable Pauli errors. The high-weight errors lead to logical coherent errors that can be suppressed exponentially by increasing the distance of the code~\cite{huang2019performance, beale2018quantum,
bravyi2018correcting, acharya_quantum_2025}, just like in the case of incoherent Pauli noise. Yet the fundamental difference between coherent errors and Pauli errors is that these high-weight errors can add coherently rather than incoherently. This can significantly alter the performance of the error correction code due to constructive interference~\cite{bravyi2018correcting, jain2023improved,marton2023coherent, venn2020errorcorrection, darmawan2024optimaladaptationsurfacecodedecoders, beale2018quantum}. Indeed, the difficulty involved in bounding the effect of interference is why there is no analytical proof of a surface code error threshold under coherent errors so far~\cite{iverson2020coherence, pato2024logical}.

Given that coherent errors can be detrimental for the performance of an error-correcting code, randomized compiling has been proposed as a way to convert coherent errors to Pauli errors~\cite{wallman2016noisetailoring}. The noise tailoring is achieved by inserting random Pauli rotations between each element of the quantum circuit. Remarkably, random Pauli rotations are also introduced under state teleportation, which motivates us in this paper to leverage these intrinsic random rotations to decohere coherent noise similar to randomized compiling. Here, we study coherent errors in teleportation-based error correction, or equivalently, measurement-based error correction (MBEC) compared to conventional QEC. In MBEC, error correction is achieved by preparing and measuring the cluster state~\cite{raussendorf_long-range_2005, raussendorf_fault-tolerant_2006, raussendorf_fault-tolerant_2007}, which can be thought of as measuring the stabilizers of some error correction code while simultaneously teleporting information onto new physical qubits~\cite{brown_universal_2020, bolt2016foliated}. We first show that repeated teleportation of a single qubit, i.e. a single-qubit teleportation chain, effectively decoheres coherent errors compared to free accumulation of errors on the qubit. The fact that teleportation-induced Pauli rotations sample only half of the single-qubit Pauli group at each time step necessitates the development of a different proof from the one used for randomized compiling. We then extend the single-qubit proof to cluster states under a physically motivated coherent error model in which single-qubit pure $Z$-coherent errors occur at all space-time locations. We find that such an error model is equivalent to incoherent Pauli errors occurring at each spacetime location. Our result bridges the gap between coherent errors and incoherent Pauli errors in MBEC. In particular, the performance of the cluster state under the circuit-level coherent error model can be efficiently simulated on a classical computer, and one can analytically prove its accuracy threshold under coherent errors by repurposing the proof for incoherent Pauli errors in topologically protected codes~\cite{DennisTopoQuantumMemory,fowler_proof_2012-1}. Indeed, for the teleported surface code under circuit-level coherent noise $e^{i\theta Z}$, we find that the threshold angle $\theta_{\mathrm{th}}$ satisfies an analytical lower bound $\theta_{\mathrm{th}}\geq\arcsin (1/10)/5$.

Our paper is structured as follows. Section~\ref{sec:coherent-noise-definition} introduces our definitions for Pauli and coherent errors. Section~\ref{sec:single-qubit-teleportation-chain} proves that errors are effectively decohered in a single-qubit teleportation chain. Section~\ref{sec:circuitlvl-eithetaz-errors} proves that circuit-level pure $Z$-coherent errors on the cluster state are exactly equivalent to physical Pauli errors. Finally, we conclude with discussions of future directions in section~\ref{sec:discussion}.

\section{Background on single-qubit error channels}
\label{sec:coherent-noise-definition}
In this section, we review single-qubit error channels. Any error channel can be written in the form $\mathcal{E}(\rho) = \sum_{k}E_k\rho E_k^\dagger$, where $E_k$ are the channel's Kraus operators, satisfying $\sum_kE_k^\dagger E_k = I$. A single-qubit Pauli error channel is defined to be 
\begin{align}
    \mathcal{E}(\rho) = \sum_{E=I,X,Y,Z}p_EE\rho E^\dagger,
\end{align}
where $p_E \geq 0$ and $\sum_Ep_E=1$. The Kraus operators of a single-qubit Pauli error channel are Pauli operators (up to a constant factor). We define a single-qubit coherent error channel to be any single-qubit channel that cannot be written in the form of a Pauli error channel. Such an error channel is coherent in the sense that at least one of its Kraus operators must be a linear combination of different Pauli operators, which results in a coherent superposition of different Pauli errors when the Kraus operator is applied on an initial state. Examples of coherent errors, as we define them, include unitary errors and amplitude damping errors.

A useful description of a single-qubit error channel $\mathcal{E}$ is the Pauli transfer matrix (PTM), which is a $4 \times 4$ matrix with elements
\begin{align}
    [\mathcal{E}]_{P,P'} = \mathrm{Tr}[P^\dagger\mathcal{E}(P')]/2,\quad P, P' \in \{I, X, Y, Z\}. \label{eq:PTM-definition}
\end{align}
Later we will allow $P, P'$ to differ from the Pauli operators by $\pm1$ or $\pm i$ for notational convenience. The channel $\mathcal{E}$ is a Pauli error channel if and only if the PTM is diagonal. Therefore, a channel is a coherent error channel by our definition if and only if its PTM has non-zero off-diagonal elements.

The distance from $\mathcal{E}$ to identity can be quantified by the average infidelity, which can be expressed in terms of the diagonal elements of the PTM \cite{nielsen_simple_2002} as follows,
\begin{align}
    r(\mathcal{E}) = \frac{1}{2}-\frac{1}{6}\left([\mathcal{E}]_{X,X} + [\mathcal{E}]_{Y,Y} + [\mathcal{E}]_{Z,Z}\right). \label{eq:average-infidelity-definition}
\end{align}
Coherent and Pauli errors can be distinguished by how the average infidelity accumulates when the error channel is applied repeatedly. In the limit of small $mr(\mathcal{E})$, the average infidelity of $m$ applications of the channel $\mathcal{E}$ accumulates linearly with $m$ if $\mathcal{E}$ is an Pauli error~\cite{iverson2020coherence}, whereas it can sometimes accumulate quadratically with $m$ if $\mathcal{E}$ is a coherent error by our definition, as is the case for a unitary rotation error.

\section{Teleportation decoheres errors on a single qubit}
\label{sec:single-qubit-teleportation-chain}
\begin{figure}[ht]
    \centering
    \includegraphics[width=\linewidth]{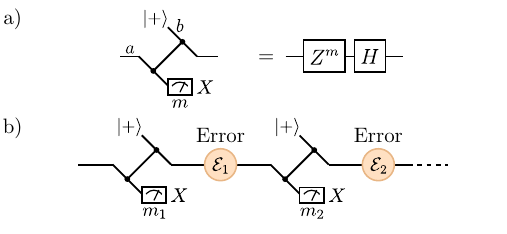}
    \caption{\textbf{The teleportation chain.} a) A one-bit teleportation circuit and its equivalent Kraus operator, up to a constant normalization factor. The qubits are labeled $a$ and $b$. The variable $m=0,1$ is the measurement outcome, and the superscript on $Z$ denotes matrix power. b) Chaining together one-bit teleportation circuits into a teleportation chain. $\mathcal{E}_t$ are single-qubit errors.}
    \label{fig:single-qubit-teleportation-chain}
\end{figure}

In this section, we prove that coherent errors in a single-qubit teleportation chain are effectively decohered. First, we define the teleportation chain and describe how it affects the Pauli frame and the errors. Next, we present a simple proof in which the errors are limited to single-qubit unitary rotation errors with rotation angles $\theta \ll 1/T$, where $T$ is the length (or equivalently the number of time steps) of the teleportation chain. While this simple proof is insufficient for constant-size errors or non-unitary errors in an arbitrarily long teleportation chain, it intuitively shows that the decoherence of errors comes from the teleportation-induced random Pauli frames, which prevent constructive interference between distant time steps. Then, we focus on what we call $Z$-like coherent errors, which include $e^{i\theta Z}$ errors, and show that they are exactly decohered into Pauli errors by teleportation. Although this proof restricts the axis of the errors, it applies to an arbitrarily long teleportation chain and holds regardless of the size of the errors. Finally, we return to arbitrary-axis single-qubit coherent errors and show that they are effectively decohered in a teleportation chain as long as their coherent parts are bounded by a small constant times their incoherent parts. We do this by first analytically bounding the error channel of the teleportation chain and then showing that the average infidelity grows at worst linearly with the length of the teleportation chain.

\subsection{Preliminary: The ideal teleportation chain}
The ideal teleportation chain, usually known as the one-dimensional cluster state \cite{brown_universal_2020}, can be formed by chaining together many copies of the one-bit teleportation circuit \cite{zhou_methodology_2000} (Fig.~\ref{fig:single-qubit-teleportation-chain}). The Kraus operator of the one-bit teleportation circuit, shown in Fig.~\ref{fig:single-qubit-teleportation-chain}a, for outcome $m=0$ or $1$ is
\begin{subequations}
\label{eq:kraus-teleportation}
\begin{align}
    K_m \!&= \begin{cases}
        \bra{+}_a \left(I_a \ket{0}_{\!b}\!\bra{0}_b + Z_a \ket{1}_{\!b}\!\bra{1}_b\right) \ket{+}_{\!b} & m=0 \\
        \bra{-}_a \left(I_a \ket{0}_{\!b}\!\bra{0}_b + Z_a \ket{1}_{\!b}\!\bra{1}_b\right) \ket{+}_{\!b} & m=1
    \end{cases} \\
    &= \begin{cases}
        \frac{1}{\sqrt{2}}(\ket{0}_{\!b}\bra{+}_a + \ket{1}_{\!b}\bra{-}_a) & m=0 \\
        \frac{1}{\sqrt{2}}(\ket{0}_{\!b}\bra{-}_a + \ket{1}_{\!b}\bra{+}_a) & m=1
    \end{cases} \\
    &= HZ^m/\sqrt{2}.
\end{align}
\end{subequations}
On the last line, the superscript $m$ in $Z^m$ denotes matrix power, and the change in the support of the state is omitted. The Kraus operator $K_m$ takes a state on qubit $a$ as input and transforms it into a state on qubit $b$. Because $K_m^\dagger K_m = I/2$, each measurement outcome is equally probable regardless of the input state. The circuit teleports the state from the input qubit to the output qubit up to a random $Z$ depending on the random outcome $m$ followed by a deterministic $H$. As the state is teleported along the single-qubit teleportation chain in Fig.~\ref{fig:single-qubit-teleportation-chain}(b), we can virtually keep track of the deterministic $H$ and the random $Z$. 

We now use calligraphic symbols $\mathcal{H}, \mathcal{I}, \mathcal{X}, \mathcal{Y}, \mathcal{Z}$ to denote the channel representation of the corresponding unitary operators $H, I, X, Y, Z$. In the ideal teleportation chain, after $t$ applications of the one-bit teleportation circuit and measuring outcomes $m_1, m_2, \cdots, m_t$, the initial state is transformed by the following unitary channel,
\begin{subequations}
\label{eq:single-qubit-teleportation-pauli-frame}
\begin{align}
    &\pf{t} = \mathcal{H}\mathcal{Z}^{m_t}\mathcal{H}\mathcal{Z}^{m_{t-1}}\cdots\mathcal{H}\mathcal{Z}^{m_1} \label{eq:single-qubit-teleportation-pauli-frame:a} \\
    &=\begin{cases}
        \mathcal{X}^{m_2 + m_4 + \cdots + m_t}\mathcal{Z}^{m_1 + m_3 + \cdots + m_{t-1}} & \text{even } t \\
        \mathcal{H}\mathcal{X}^{m_2 + m_4 + \cdots + m_{t-1}}\mathcal{Z}^{m_1 + m_3 + \cdots + m_t} & \text{odd } t
    \end{cases}. \label{eq:single-qubit-teleportation-pauli-frame:b}
\end{align}
\end{subequations}
As an example of the notation, if $m_1 + m_3 = 1$ and $m_2 = 1$, then 
\begin{align*}
    \pf{3}(\rho) = (\mathcal{H}\mathcal{X}\mathcal{Z})(\rho) = HXZ\rho ZXH.
\end{align*}
Note that the $\mathcal{X}$ emerges in $\pf{t}$ because of the conjugation of $\mathcal{Z}$ by $\mathcal{H}$. The unitary channel $\pf{t}$ can be thought of as a frame transformation that flips and permutes the Pauli operators $X, Y, Z$, so we refer to $\pf{t}$ as the Pauli frame for short. As usual, at the end of a teleportation chain, the inverse $\pf{t}^{-1}$ is applied either physically or virtually to recover the original state.

From the above discussion, we see that the Pauli frame is randomized every time the one-bit teleportation circuit is applied. This is similar to randomized compiling \cite{wallman2016noisetailoring}, which achieves the randomization via the deliberate insertion of random single-qubit Pauli gates into the circuit. Because each measurement outcome $m_t=0,1$ is random, Eq.~\eqref{eq:single-qubit-teleportation-pauli-frame:b} shows that each Pauli frame $\pf{t}$ is uniformly distributed over its possible values,
\begin{subequations}
\label{eq:single-qubit-teleportation-pauli-frame-distribution}
\begin{align}
    \pf{t} &= \mathcal{H}, \mathcal{H}\mathcal{Z}  &\text{if } t=1, \label{eq:single-qubit-teleportation-chain-uniform-marginal-distribution-one} \\
    \pf{t} &= \mathcal{I}, \mathcal{X}, \mathcal{Y}, \mathcal{Z}  &\text{if even }t, \label{eq:single-qubit-teleportation-chain-uniform-marginal-distribution-even} \\
    \pf{t} &= \mathcal{H}, \mathcal{H}\mathcal{X}, \mathcal{H}\mathcal{Y}, \mathcal{H}\mathcal{Z}  &\text{if odd }t>1. \label{eq:single-qubit-teleportation-chain-uniform-marginal-distribution-odd}
\end{align} \label{eq:single-qubit-teleportation-chain-uniform-marginal-distribution}
\end{subequations}
We now observe a characteristic of the teleportation chain that will soon become important: despite having uniform distributions at a fixed time, these Pauli frames have a short-lived correlation over time. The existence of correlation can be seen, for example, from the inequality of the following conditional and unconditional probabilities:
\begin{align*}
    &\Pr(\pf{2}=\mathcal{Z}|\pf{1} = \mathcal{H}) = 0, \\
    &\Pr(\pf{2}=\mathcal{Z}) = 1/4.
\end{align*}
More generally, the Pauli frame $\pf{t}$ is correlated with the Pauli frame $\pf{t-1}$. The correlation distinguishes the distribution of Pauli frames due to teleportation from that due to randomized compiling in which the Pauli frames at two consecutive time steps are completely uncorrelated. However, the correlations in the teleportation chain are short-lived. For example, from Eq.~\eqref{eq:single-qubit-teleportation-pauli-frame:a}, 
\begin{align*}
    \pf{3} &= \mathcal{H}\mathcal{Z}^{m_3}\mathcal{H}\mathcal{Z}^{m_2}\pf{1} \\
    &= \mathcal{X}^{m_3}\mathcal{Z}^{m_2}\pf{1},
\end{align*}
so no matter $\pf{1} = \mathcal{H}$ or $\mathcal{H}\mathcal{Z}$, $\pf{3}$ always takes values $\mathcal{H}, \mathcal{H}\mathcal{X}, \mathcal{H}\mathcal{Y}, \mathcal{H}\mathcal{Z}$ with uniform probability. More generally, the Pauli frames $\pf{s}$ and $\pf{t}$ are independent as long as $|s-t| > 1$. This situation resembles a version of randomized compiling that is performed every two time steps. We will see that the short correlation time is important for the decoherence of errors later.

\subsection{Teleportation chain with errors}

Now consider the teleportation chain with errors. We restrict to errors that can be lumped into single-qubit error channels $\mathcal{E}_t$ between teleportations (Fig.~\ref{fig:single-qubit-teleportation-chain}b). Examples include arbitrary idle errors between teleportations and state preparation and measurement errors that commute with the CZ gates. Another example is a CZ gate error in the form of unknown single-qubit $e^{i\theta Z}$ rotations on the control and target qubits. This type of error can for instance happen on the neutral atom platform due to spurious frequency shifts and pulse miscalibrations \cite{levine_parallel_2019, bluvstein_quantum_2022}.

Because we always apply $\pf{t}^{-1}$ at the end of the teleportation chain, we are essentially working in the interaction picture defined by $\pf{t}$. In this interaction picture, a physical error $\mathcal{E}_t$ at time $t$ appears as the error
\begin{align}
    \mathcal{E}_t' = \pf{t}^{-1}\mathcal{E}_t\,\pf{t}. \label{eq:single-qubit-teleportation-transformed-error}
\end{align}
The transformed error changes with the Pauli frame $\pf{t}$, which in turn depends on the random measurement outcomes $m_1, m_2, \cdots, m_t$. The error channel of the teleportation chain conditioned on measurement outcomes $m_1, m_2, \cdots, m_t$ is just the product of all such transformed errors,
\begin{align}
    \mathcal{N}_t = \mathcal{E}_t'\mathcal{E}_{t-1}' \cdots \mathcal{E}_1'. \label{eq:single-qubit-teleportation-cumulative-error}
\end{align}
Thus we see that the total error channel $\mathcal{N}_t$ changes with the Pauli frames $\pf{t}$ and the measurement outcomes $m_1, m_2, \cdots, m_t$. We denote the average of $\mathcal{N}_t$ over all $m_1, m_2, \cdots, m_t$ (or equivalently over the distribution of Pauli frame sequences) to be
\begin{align}
    \overline{\mathcal{N}_t} = \overline{\mathcal{E}_t'\mathcal{E}_{t-1}' \cdots \mathcal{E}_1'}. \label{eq:single-qubit-teleportation-average-error-channel}
\end{align}
Note that even though the outcomes $m_1, m_2, \cdots m_t$ are known in an experiment and indeed are used to determine the final Pauli frame, the average error channel $\overline{\mathcal{N}_t}$ is relevant as long as no further action is taken conditioned on this knowledge.

Recall from the previous subsection that the Pauli frames in the teleportation chain are correlated. If the Pauli frames were completely uncorrelated (as they are in the case of randomized compiling performed at every timestep), we would be able to use the following argument to show that the errors are decohered in the teleportation chain. As a consequence of the uncorrelated Pauli frames, the transformed errors $\mathcal{E}_t'$ would also be uncorrelated, which would allow the factorization of the average $\overline{\mathcal{N}_t}$ into a product of individual averages of $\mathcal{E}_t'$,
\begin{align*}
    \overline{\mathcal{N}_t} &= \overline{\mathcal{E}_t'}\ \overline{\mathcal{E}_{t-1}'}\cdots\overline{\mathcal{E}_1'}, \\
    \text{where } \overline{\mathcal{E}_t'} &= \sum_{\pf{t}} \Pr(\pf{t})\ \pf{t}^{-1}\mathcal{E}_t\pf{t}.
\end{align*}
Substituting the uniform distribution of $\pf{t}$ from Eq.~\eqref{eq:single-qubit-teleportation-pauli-frame-distribution} would show that the average $\overline{\mathcal{E}_t'}$ for $t > 1$ is the \textit{Pauli twirl} \cite{silva_scalable_2008, dankert_exact_2009, emerson_symmetrized_2007, divincenzo_quantum_2002} of either $\mathcal{E}_t$ or $\mathcal{H}\mathcal{E}_t\mathcal{H}$,
\begin{align}
    \overline{\mathcal{E}_t'}&= \begin{cases}
    \frac{1}{4}\sum_{\mathcal{P}=\mathcal{I}, \mathcal{X}, \mathcal{Y}, \mathcal{Z}} \mathcal{P}\mathcal{E}_t\mathcal{P} & \text{even } t \\
    \frac{1}{4}\sum_{\mathcal{P}=\mathcal{I}, \mathcal{X}, \mathcal{Y},\mathcal{Z}} \mathcal{P}\mathcal{H}\mathcal{E}_t\mathcal{H}\mathcal{P} & \text{odd } t > 1 \\
    \end{cases}, \label{eq:pauli-twirl}
\end{align}
which would remove the coherence and result in a Pauli error channel. Therefore, each coherent error $\mathcal{E}_t$ for $t>1$ would be decohered into a Pauli error. However, since the Pauli frames in the teleportation chain are correlated, the argument above is invalid. Fortunately, the correlation between Pauli frames is short-lived, as discussed in the previous subsection. The short-livedness of correlation allows us to show the decoherence of error using the simple but approximate argument presented in the next subsection.

\subsection{Simple proof of the decoherence of errors}
\label{sec:simple-proof}
In this subsection, we present a simple proof that shows that, despite the correlation of Pauli frames, single-qubit unitary rotation errors $e^{i\boldsymbol{\theta}_t \cdot \boldsymbol{\sigma}}$ around arbitrary axes are effectively decohered by the teleportation chain as long as the size of the rotation angles $|\boldsymbol{\theta}_t|$ are bounded by some $\theta \ll 1/T$, where $T$ is the length of the teleportation chain.

Let each $\mathcal{E}_t$ be a single-qubit unitary rotation error $e^{i\boldsymbol{\theta}_t \cdot \boldsymbol{\sigma}}$ around some arbitrary axis, and let $\theta$ upper bound all rotation angles $\boldsymbol{\theta}_t$. By Eq.~\eqref{eq:single-qubit-teleportation-transformed-error}, the effective error is also a unitary rotation whose rotation angle $\boldsymbol{\theta}_t'$ relates to $\boldsymbol{\theta}_t$ by
\begin{align}
    \boldsymbol{\theta}_t' \cdot \boldsymbol{\sigma} = \pf{t}^{-1}\!\left(\boldsymbol{\theta}_t \cdot \boldsymbol{\sigma}\right). \label{eq:angle-transformation}
\end{align}
The expression for the average error channel in this case is
\begin{align}
    \overline{\mathcal{N}_t}(\rho) &= \overline{\left(\prod_{s=1}^te^{i\boldsymbol{\theta}_s' \cdot \boldsymbol{\sigma}}\right)\rho\left(\prod_{s=1}^te^{i\boldsymbol{\theta}_s' \cdot \boldsymbol{\sigma}}\right)^\dagger}.
\end{align}
We can use a rotation by the net angle $\boldsymbol{\Theta} = \sum_{s=1}^t\boldsymbol{\theta}_s'$ (satisfying $||\boldsymbol{\Theta}|| \leq \theta t$) to approximate the product of small rotations, that is,
\begin{align}
    \prod_{s=1}^te^{i\boldsymbol{\theta}_s' \cdot \boldsymbol{\sigma}} \approx e^{i\boldsymbol{\Theta} \cdot \boldsymbol{\sigma}}. \label{eq:exponential-approximation}  
\end{align}
In the exponent on the right-hand side, the approximation omits $O(t^2)$ second-order commutator terms of the form $[\boldsymbol{\theta}_{s_1}' \cdot \boldsymbol{\sigma}, \boldsymbol{\theta}_{s_2}' \cdot \boldsymbol{\sigma}]$, each of which is of magnitude $O(\theta^2)$. So the right-hand side is only valid for $\theta t \ll 1$. This means that we require $\theta \ll 1/T$ for a teleportation chain of length $T$, or equivalently, we require a short-enough teleportation chain $T \ll 1 / \theta$ for a constant $\theta$. Inserting the approximate Eq.~\eqref{eq:exponential-approximation} back into $\overline{\mathcal{N}_t}$ and keeping terms to second order in $\boldsymbol{\Theta}$ gives the following expression,
\begin{align}
    \overline{\mathcal{N}_t}(\rho) &\approx \left(1 - \sum_{\mathclap{k=x,y,z}}\overline{\Theta_k^2}\right)\, I \rho I + \sum_{\mathclap{k=x,y,z}}\overline{\Theta_k^2}\  \sigma_k\rho \sigma_k \nonumber \\
    &\quad + \sum_{\mathclap{k=x,y,z}}\left(i\overline{\Theta_k}\ \sigma_k \rho I + \mathrm{h.c.}\right) + \sum_{k\neq l}\overline{\Theta_k\Theta_l}\ \sigma_k\rho \sigma_l.
\end{align}
Note that the coefficients of $I\rho I$ and $\sigma_k\rho\sigma_k$ are accurate to second order in $\theta^2$ despite the first-order approximate Eq.~\eqref{eq:exponential-approximation}. To see this, first observe that $\boldsymbol{\Theta} = O(\theta t)$, and the omitted commutator terms only leads to $\boldsymbol{\Theta} \to \boldsymbol{\Theta} + O(\theta^2t^2)$. Plugging this into the coefficients of $I\rho I$ and $\sigma_k\rho\sigma_k$ shows that the omitted commutator terms lead to corrections that are third-order in $\theta$.

Let's examine the Pauli part and the coherence of $\overline{\mathcal{N}_t}$. The Pauli error probabilities can be identified as $\overline{\Theta_k^2}$. If we now plug in $\Theta_k = \sum_{s=1}^t\theta_{k,s}'$, we get
\begin{align}
    \overline{\Theta_k^2} &= \sum_{s=1}^t\overline{(\theta_{k,s}')^2} + 2\sum_{s=1}^t\sum_{r=s+1}^t\overline{\theta_{k,s}'\theta_{k,r}'}.
\end{align}
Because of the short correlation time of $\pf{t}$, $\theta_{k,s}'$ and $\theta_{k,r}'$ are independent for $|s-r| > 1$. Therefore, for $r > s+1$, the two-time average $\overline{\theta_{k,s}'\theta_{k,r}'}$ factors into $\overline{\theta_{k,s}'}\ \overline{\theta_{k,r}'}$. Furthermore, we can show $\overline{\theta_{k,r}'} = 0$ for $r > 1$ by averaging Eq.~\eqref{eq:angle-transformation} over the uniform distribution in Eq.~\eqref{eq:single-qubit-teleportation-pauli-frame-distribution}. As a result, $\overline{\theta_{k,s}'\theta_{k,r}'} = 0$ for $r > s+1$, and only $O(t)$ terms under the double sum survive the average,
\begin{subequations}
\begin{align}
    \overline{\Theta_k^2} &= \sum_{s=1}^t\overline{(\theta_{k,s}')^2} + 2\sum_{s=1}^{t-1}\overline{\theta_{k,s}'\theta_{k,s+1}'} \\
    &=\overline{(\theta_{k,1}')^2} + \sum_{s=1}^t \left[\overline{(\theta_{k,s}')^2} + \overline{\theta_{k,s-1}'\theta_{k,s}'}\right]. \label{eq:fast-and-loose-pauli-error-prob}
\end{align}
\end{subequations}
Similar reasoning can be used to show that $\overline{\Theta_k\Theta_l}$ is at most $O(t\theta^2) \ll \theta$, so the coherence of $\overline{\mathcal{N}_t}$ is dominated by the term $i\overline{\Theta_k}\ \sigma_k \rho I + \mathrm{h.c.}$, which is first order in $\theta$. Because we have shown that $\overline{\theta_{k,r}'} = 0$ for $r > 1$, we must have $\overline{\Theta_k} = \overline{\theta_{k,1}'} = O(\theta)$, which does not grow with the length of the teleportation chain.

Our observations of the Pauli part and the coherence in $\overline{\mathcal{N}_t}$ can be used to show that the errors are effectively decohered in the short teleportation chain ($T \ll 1/\theta$). First, while $\overline{\mathcal{N}_t}$ is a coherent error channel by our definition in section \ref{sec:coherent-noise-definition}, its coherence is small, bounded by $O(\theta)$ and does not grow with the number of timesteps. This coherence comes from the coherent error at $t=1$, which does not get completely decohered under the insufficiently random Pauli frame at $t=1$ [See Eq.~\eqref{eq:single-qubit-teleportation-chain-uniform-marginal-distribution-one}]. Next, from Eq.~\eqref{eq:fast-and-loose-pauli-error-prob}, we see that to second order in $\theta$ the Pauli part of $\overline{\mathcal{N}_t}$ can be rewritten as the product of a sequence of Pauli channels. The Pauli error probabilities of these Pauli channels are $p_{k,t} = \overline{(\theta_{k,t}')^2} + \overline{\theta_{k,t-1}'\theta_{k,t}'}$, with the exception of $t=1$. The first part of $p_{k,t}$ comes simply from the Pauli error part of the original error channel. The second part of $p_{k,t}$ accounts for interference that can happen within the correlation time of Pauli frames in the teleportation chain. Combining the Pauli part and the coherence then shows that $\overline{\mathcal{N}_t}$ can be written as the product of a small coherent error channel at $t=1$ followed by Pauli error channels that occur at all subsequent timesteps, showing that the physical errors have effectively been decohered.

From the expression for $\overline{\mathcal{N}_t}$ above, we can additionally compute the average infidelity of the teleportation chain as the chain grows longer. Applying the formula for average infidelity from Eq.~\eqref{eq:average-infidelity-definition} to our expression for $\overline{\mathcal{N}_t}$ gives
\begin{align}
    r(\overline{\mathcal{N}_t}) &\approx \!\frac{2}{3}\overline{||\boldsymbol{\Theta}||^2} = \frac{2}{3}\sum_{s=1}^t||\boldsymbol{\theta}_s||^2 + \frac{4}{3}\sum_{s=1}^{t-1}\overline{\boldsymbol{\theta}_s'\cdot \boldsymbol{\theta}_{s+1}'}.
    \label{eq:fast-and-loose-result}
\end{align}
If we let $r_0 \ll 1/T^2$ upper bound the average infidelity of each physical error channel, then $\theta \lesssim \sqrt{3r_0/2}$. Substituting the bound for the physical rotation angle into the expression for $r(\overline{\mathcal{N}_t})$ gives the approximate upper bound
\begin{align}
    r(\overline{\mathcal{N}_t}) \lesssim 3r_0t. \label{eq:worst-case-infidelity-simple-proof}
\end{align}
This expression shows that the average infidelity accumulates at worst linearly with $t$. The simple proof in this section motivates that an arbitrarily long teleportation chain decoheres errors despite the correlation in the Pauli frames. We indeed find this to be the case with more rigorous proofs in the following sections.

\subsection{Proof for \texorpdfstring{$e^{i\theta Z}$}{exp(i theta Z)} errors and \texorpdfstring{$Z$}{z}-like coherent errors}
\label{sec:z-like-coherence}
Now we show that if $\mathcal{E}_t$ are $e^{i\theta Z}$ errors, or more generally, what we call $Z$-like coherent errors, then the teleportation chain exactly converts each $\mathcal{E}_t$ into its Pauli twirl [in the sense of Eq.~\eqref{eq:pauli-twirl}]. The proof applies to an arbitrary long teleportation chain regardless of the size of the errors. Our proof will utilize the structure we impose on the errors and the details of the teleportation chain, as opposed to just the short correlation time of the Pauli frames. 

To define $Z$-like coherent errors, we first observe that any error channel $\mathcal{E}_t$ can be decomposed into the sum of its Pauli component $\mathcal{E}_{I,t}$ and its coherent components $\mathcal{E}_{X,t}$, $\mathcal{E}_{Y,t}$, $\mathcal{E}_{Z,t}$,
\begin{align}
    \mathcal{E}_t &= \mathcal{E}_{I,t} + \mathcal{E}_{X,t} + \mathcal{E}_{Y,t} + \mathcal{E}_{Z,t}. \label{eq:coherence-decomposition}
\end{align}
These components are defined by their sign changes under Pauli conjugation,
\begin{align}
    &\mathcal{X}^{x}\mathcal{Z}^{z}\ \mathcal{E}_t\ \mathcal{X}^{x}\mathcal{Z}^{z} \nonumber \\
    &= \mathcal{E}_{I,t} + (-1)^{z}\mathcal{E}_{X,t} + (-1)^{x+z}\mathcal{E}_{Y,t} + (-1)^{x}\mathcal{E}_{Z,t}. \label{eq:coherence-transformation}
\end{align}
In particular, the incoherent component $\mathcal{E}_{I,t}$ is the Pauli twirl of $\mathcal{E}_t$. A \textit{$Z$-like coherent error} is an error channel $\mathcal{E}_t$ where $\mathcal{E}_{X,t} = \mathcal{E}_{Y,t} = 0$, and therefore admits decomposition
\begin{align}
    \mathcal{E}_t = \mathcal{E}_{I,t} + \mathcal{E}_{Z,t}. \label{eq:z-like-coherent-error}
\end{align}
An example is a unitary rotation around the $z$-axis, i.e., $\mathcal{E}_t(\rho) = e^{i\theta Z}\rho\,e^{-i\theta Z}$. Its Pauli component is the Pauli error channel $\mathcal{E}_{I, t}(\rho) = \cos^2\!\theta\,I \rho I + \sin^2\!\theta\,Z \rho Z$, and the only non-zero coherent component is $\mathcal{E}_{Z, t}(\rho) = -i\cos\theta\sin\theta\,(I \rho Z -Z \rho I)$.

To derive the average error channel of the teleportation chain under $Z$-like coherent errors, we return to the expression of the transformed errors $\mathcal{E}_t'$ from Eq.~\eqref{eq:single-qubit-teleportation-transformed-error} and substitute in Eq.~\eqref{eq:z-like-coherent-error},
\begin{subequations}
\begin{align}
    &\mathcal{E}_t' = \pf{t}^{-1}(\mathcal{E}_{I, t} + \mathcal{E}_{Z, t})\pf{t} \\
    &= \begin{cases}
        \mathcal{E}_{I, t} + (-1)^{m_2 + m_4 + \cdots + m_t}\mathcal{E}_{Z, t} & \text{even }t \\
        \mathcal{H}\left[\mathcal{E}_{I, t} + (-1)^{m_1 + m_3 + \cdots + m_t}\mathcal{E}_{Z, t}\right]\mathcal{H} & \text{odd }t \\
    \end{cases}. \label{eq:transformed-z-like-coherent-error}
\end{align}
\end{subequations}
In obtaining Eq.~\eqref{eq:transformed-z-like-coherent-error}, we first substitute Eq.~\eqref{eq:single-qubit-teleportation-pauli-frame:b} for $\pf{t}$ and then use Eq.~\eqref{eq:coherence-transformation} to determine the sign flips of the $Z$-coherence. Eq.~\eqref{eq:transformed-z-like-coherent-error} shows that $\mathcal{E}_t'$ is an error channel that depends only on the exponent of $(-1)$, which is a random bit
\begin{align*}
    q_t &= \begin{cases}
        (m_2 + m_4 + \cdots + m_t)\!\!\!\!\mod 2 & \text{even }t \\
        (m_1 + m_3 + \cdots + m_t)\!\!\!\!\mod 2 & \text{odd }t
    \end{cases}.
\end{align*}
The bits $q_1, q_2, \cdots, q_t$ are again independent random bits by construction, so the average channel factorizes into
\begin{align*}
    \overline{\mathcal{N}_t} &= \overline{\mathcal{E}_t'}\ \overline{\mathcal{E}_{t-1}'}\cdots\overline{\mathcal{E}_1'}. \\ 
    \text{where } \overline{\mathcal{E}_t'} &= \begin{cases}
        \mathcal{E}_{I, t} + \overline{(-1)^{q_t}}\mathcal{E}_{Z, t} & \text{even }t \\
        \mathcal{H}\left[\mathcal{E}_{I, t} + \overline{(-1)^{q_t}}\mathcal{E}_{Z, t}\right]\mathcal{H} & \text{odd }t \\
    \end{cases} \\
    &= \begin{cases}
        \mathcal{E}_{I, t} & \text{even }t \\
        \mathcal{H}\mathcal{E}_{I, t}\mathcal{H} & \text{odd }t \\
    \end{cases}
\end{align*}
Therefore, teleportation decoheres $Z$-like coherent errors into Pauli errors. We remark that this proof can be modified to show that teleportation decoheres $X$-like coherent errors and $Y$-like coherent errors.

\subsection{Proof for arbitrary-axis single-qubit errors}
In this subsection, we return to arbitrary-axis single-qubit coherent errors and show that they are effectively decohered in an arbitrarily long teleportation chain as long as their coherent parts are bounded by a small constant times their incoherent parts. First, we present analytical bounds on the diagonal elements of the Pauli transfer matrix of the teleportation chain. Much like the simple proof, these bounds show that the error channel of the single-qubit teleportation chain can be written as the product of a single coherent error channel followed by Pauli errors occurring at all subsequent timesteps. Next, we numerically verify the bounds by comparing it to the exact error channel of the teleportation chain. We empirically observe that the average fidelity of the exact error channel as well as the corresponding bounds grow linearly with time, in contrast to growing quadratically in the absence of teleportation. Finally, we use the analytical bounds to derive an expression for the average infidelity of the teleportation chain similar to Eq.~\eqref{eq:worst-case-infidelity-simple-proof}, proving that it grows at most linearly with $t$ regardless of the length of the teleportation chain.

The matrix elements of the PTM are related to the error channel's incoherent and coherent components as defined by Eqs.~\eqref{eq:coherence-decomposition}~and~\eqref{eq:coherence-transformation}. By observing the sign flips of $[\mathcal{X}^x\mathcal{Z}^z\mathcal{E}\mathcal{X}^x\mathcal{Z}^z]_{P,P'}$, we can deduce that the matrix elements $[\mathcal{E}]_{P,P'}$ where $P=P'$ contribute to the channel's Pauli part, and the matrix elements where $P$ and $P'$ differ by $X$, $Y$, or $Z$ contribute to the channel's $X$-coherence, $Y$-coherence, and $Z$-coherence, respectively. Let us also absorb the Hadarmard occurring at odd $t$ into the error channel $\mathcal{E}_t$ using the notation 
\begin{align}
    \mathcal{E}_t^H = \begin{cases}
        \mathcal{E}_t & \text{even }t \\
        \mathcal{H}\mathcal{E}_t\mathcal{H} & \text{odd }t
    \end{cases}.
\end{align}

Using the above notations, we present a theorem that analytically bounds the error channel of the teleportation chain when the coherent parts of the physical errors are bounded by a small constant times their incoherent parts:
\begin{widetext}
    \begin{theorem}
        Let $\left[\mathcal{P}_{(t,t-1)}\right]_{P,P}$ denote the relative change in the diagonal element $\left[\overline{\mathcal{N}_t}\right]_{P,P}$ of the PTM of the average error channel of the teleportation chain from $t-1$ to $t$, that is,
        \begin{align}
            \left[\overline{\mathcal{N}_t}\right]_{P,P} = \left[\mathcal{P}_{(t,t-1)}\right]_{P,P}\left[\overline{\mathcal{N}_{t-1}}\right]_{P,P}. \label{eq:theorem-1-relative-change-definition}
        \end{align}
        If the coherences $\mathcal{E}_{X,t}, \mathcal{E}_{Y,t}, \mathcal{E}_{Z,t}$ in the error channels $\mathcal{E}_t$ in the teleportation chain are at most $\epsilon$ times the size of the incoherent components $\mathcal{E}_{I,t}$ with $\epsilon < 1/3$, or more precisely,
        \begin{align}
            \epsilon = \max_t\max_{P\in\{I,X,Y,Z\}}\left(
            \left|\frac{\left[\mathcal{E}_t\right]_{P,XP}}{\left[\mathcal{E}_t\right]_{P,P}}\right|, \left|\frac{\left[\mathcal{E}_t\right]_{ZP,XP}}{\left[\mathcal{E}_t\right]_{P,P}}\right|, \left|\frac{\left[\mathcal{E}_t\right]_{ZP,P}}{\left[\mathcal{E}_t\right]_{P,P}}\right|
            \right) < \frac{1}{3}, \label{eq:theorem-1-smallness-condition}
        \end{align}
        then the coherence in $\overline{\mathcal{N}_t}$ is $O(\epsilon)$ and for all $P$ and $t > 2$,
        \begin{align}
            &\left[\mathcal{P}_{(t,t-1)}\right]_{P,P} = 
            (1 + \delta_1)\left[\mathcal{E}_t^H\right]_{P,P}
            + (1 + \delta_2)
            \begin{cases}
                \left[\mathcal{E}_t^H\right]_{P,XP}\left[\mathcal{E}_{t-1}^H\right]_{XP,P} \Big/ \left[\mathcal{E}_{t-1}^H\right]_{P,P} & \text{even }t \\
                \left[\mathcal{E}_t^H\right]_{P,ZP}\left[\mathcal{E}_{t-1}^H\right]_{ZP,P} \Big/\left[\mathcal{E}_{t-1}^H\right]_{P,P} & \text{odd }t
            \end{cases}, \label{eq:single-qubit-teleportation-chain-iterative-bounds} \\
            &\text{for some } \delta_1 \in \left[-\frac{3\epsilon^3}{1 - 3\epsilon^2}, +\frac{3\epsilon^3}{1 - 3\epsilon^2}\right] \text{ and } \delta_2 \in \left[-\frac{3\epsilon^2}{1 + 3\epsilon^2}, +\frac{3\epsilon^2}{1 - 3\epsilon^2}\right]. \nonumber
        \end{align}
    \label{thm:bound-on-error-growth-per-timestep}
    \end{theorem}
\end{widetext}
Similar to Eq.~\eqref{eq:worst-case-infidelity-simple-proof} in the simple proof, the expression for $\left[\mathcal{P}_{(t,t-1)}\right]_{P,P}$ also has two parts. Ignoring the small corrections $\delta_1 =O(\epsilon^3)$ and $\delta_2 =O(\epsilon^2)$, the first term in $\left[\mathcal{P}_{(t,t-1)}\right]_{P,P}$ is simply the Pauli part of $\mathcal{E}_t$ (or $\mathcal{H}\mathcal{E}_t\mathcal{H}$ for odd $t$), which would have resulted if the errors in the teleportation chain are converted into Pauli errors via randomized compiling. The second term in $\left[\mathcal{P}_{(t,t-1)}\right]_{P,P}$ contains a product of the coherence parts of neighboring error channels, which accounts for interference within the short correlation time of the Pauli frame. 

The second term in Eq.~\eqref{eq:single-qubit-teleportation-chain-iterative-bounds} for $\left[\mathcal{P}_{(t,t-1)}\right]_{P,P}$ is $O(\epsilon^2)$ by the smallness condition on the coherences defined by Eq.~\eqref{eq:theorem-1-smallness-condition}, whereas $\delta_1 = O(\epsilon^3)$ and $\delta_2 = O(\epsilon^2)$, so Eq.~\eqref{eq:single-qubit-teleportation-chain-iterative-bounds} for $\left[\mathcal{P}_{(t,t-1)}\right]_{P,P}$ is accurate to second order in $\epsilon$. We therefore refer to Eq.~\eqref{eq:single-qubit-teleportation-chain-iterative-bounds} as the second-order bounds on the error growth per timestep in the teleportation chain. We provide a proof for theorem~\ref{thm:bound-on-error-growth-per-timestep} in appendix \ref{app:bounds-proof}. In appendix \ref{sec:tighter-bounds-proof} we also prove an expression for the analogous quantity $\left[\mathcal{P}_{(t,t-2)}\right]_{P,P}$, the relative change between $\left[\overline{\mathcal{N}_t}\right]_{P,P}$ and $\left[\overline{\mathcal{N}_{t-2}}\right]_{P,P}$, that is accurate to third order in $\epsilon$.

Just as the simple proof, theorem~\ref{thm:bound-on-error-growth-per-timestep} shows that the coherence of $\overline{\mathcal{N}_t}$, while non-zero, is small and bounded by a constant, and that the Pauli part of $\overline{\mathcal{N}_t}$ can be written as product of Paul error channels. We can therefore write $\overline{\mathcal{N}_t}$ as the product of a single coherent error channel to capture the coherence, followed by a sequence of Pauli error channels $\mathcal{P}_{(t, t-1)}$ to capture the Pauli error part. Therefore, the errors in an arbitrarily long single-qubit teleportation chain is effectively decohered.

\begin{figure}[th]
    \centering
    \includegraphics[width=\linewidth]{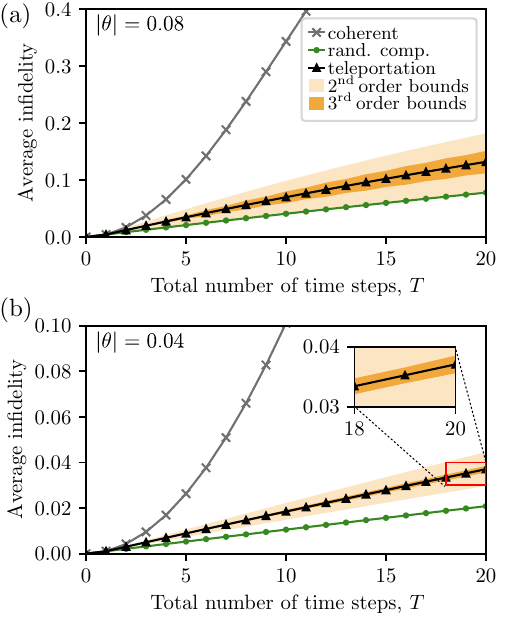}
    \caption{\textbf{Decoherence of errors in a single-qubit teleportation chain.} We compare the error growth over time steps due to physical coherent errors in three scenarios (see main text for details): in a teleportation chain (black triangles), under randomized compiling (green dots), and when neither randomized compiling nor teleportation is performed (gray crosses). We fix the physical error per time step to be $\mathcal{E}_t(\rho) = e^{-i\boldsymbol{\theta} \cdot \boldsymbol{\sigma}} \rho e^{i\boldsymbol{\theta} \cdot \boldsymbol{\sigma}}$, where $\boldsymbol{\theta} = |\theta| \left(3\mathbf{\hat{x}} + 1\mathbf{\hat{y}} + 2\mathbf{\hat{z}}\right)/\sqrt{14}$ and (a) $|\theta| = 0.08$, (b) $|\theta| = 0.04$. The vertical axis in each subplot is the average infidelity defined in Eq.~\eqref{eq:average-infidelity-definition}. We also superimpose the second- and third-order bounds from theorem~\ref{thm:bound-on-error-growth-per-timestep} and from the appendix \ref{sec:tighter-bounds-proof}. The growth or errors with time steps appears linear under teleportation and under randomized compiling, indicating the decoherence of errors. The inset in (b) zooms in to the regime outline by the red rectangle.}
    \label{fig:analytical-bounds}
\end{figure}

In Fig.~\ref{fig:analytical-bounds}, we verify the second- and third-order bounds on error growth in a teleportation chain by comparing them to exact numerics. We choose $\mathcal{E}_t(\rho) = e^{-i\boldsymbol{\theta}_t \cdot \boldsymbol{\sigma}} \rho e^{i\boldsymbol{\theta}_t \cdot \boldsymbol{\sigma}}$ where $\boldsymbol{\theta}_t$ is a constant and equals $|\theta| \left(3\mathbf{\hat{x}} + 1\mathbf{\hat{y}} + 2\mathbf{\hat{z}}\right)/\sqrt{14}$. The average error channel $\overline{\mathcal{N}_t}$ and the corresponding average infidelity $r(\overline{\mathcal{N}_t})$ can be computed exactly using Eqs.~\eqref{eq:single-qubit-teleportation-chain-conditional-average-channel-iterative}~and~\eqref{eq:single-qubit-teleportation-chain-average-conditional-channel} derived in appendix \ref{app:bounds-proof}. Second-order upper and lower bounds on $r(\overline{\mathcal{N}_t})$ can be computed from theorem~\eqref{thm:bound-on-error-growth-per-timestep} Eq.~\eqref{eq:single-qubit-teleportation-chain-iterative-bounds} by taking appropriate combinations of extremum values for $\delta_1$ and $\delta_2$, multiplying together all $\left[\mathcal{P}_{(t,t-1)}\right]_{P,P}$, and using Eq.~\eqref{eq:average-infidelity-definition} to find the infidelity. The third-order bounds can be computed similarly. Both the second- and third-order bounds agree with the exact numerical average infidelity, with the third-order bounds having lower uncertainty as expected.

In Fig.~\ref{fig:analytical-bounds}, we also compare the growth of average infidelity as a function of time in the teleportation chain to two other scenarios. In the first scenario, the errors $\mathcal{E}_t$ are applied on a single qubit and accumulates freely (gray curve). For the chosen errors, the average infidelity shows a quadratic increase with time characteristic of constructive interference between errors at different time steps. In the second scenario, each error $\mathcal{E}_t$ is first Pauli twirled into its incoherent part $\mathcal{E}_{I,t}$ via randomized compiling and then applied on a single qubit. In both teleportation and randomized compiling, the growth of average infidelity appears linear with time, indicating the decoherence of the errors. This indicates that teleportation effectively decoheres coherent errors into Pauli errors.

To rigorous support our empirical observation that the average infidelity grows linearly with time in a teleportation chain, we prove the following corollary using theorem~\ref{thm:bound-on-error-growth-per-timestep}:
\begin{corollary}
    \label{cor:worst-case-infidelity-growth}
    If the average infidelity of each $\mathcal{E}_t$ is at most $r_0 \leq 1/100$, then the average infidelity of the teleportation chain grows at most linearly with time, specifically,
    \begin{align}
        r(\overline{\mathcal{N}_t}) \leq \frac{17}{2}r_0t. \label{eq:worst-case-infidelity}
    \end{align}
\end{corollary}
\!\!\!\!\!\!Note that the coefficient $17/2$ can be reduced because it depends on the constraint $r_0 < 1/100$ and the order of the bounds on $\left[\mathcal{P}_{(t, t-1)}\right]_{P,P}$ used to prove the corollary. While Eq.~\eqref{eq:worst-case-infidelity} in corollary \ref{cor:worst-case-infidelity-growth} is a looser bound than Eq.~\eqref{eq:worst-case-infidelity-simple-proof} in the simple proof, the former has a wider range of applicability. Recall that for the derivation of Eq.~\eqref{eq:fast-and-loose-result} to hold, we require $t < O(1/\theta)$, where $\theta$ is the rotation angle of the physical unitary rotation. Here, our corollary holds for any single-qubit errors $\mathcal{E}_t$ without any restriction on the length of the teleportation chain as long as the average infidelity of each $\mathcal{E}_t$ is no more than a small constant $1/100$.

The proof of the corollary is as follows. If $\mathcal{E}_t$ has average infidelity $r_0\leq 1/100$, then we have the following bounds \cite{beale2018quantum} on the elements of the PTM,
\begin{align*}
    \left[\mathcal{E}_t\right]_{P,P} \geq 1 - 3r_0 > 0, \quad \left|\left[\mathcal{E}_t\right]_{P,P'}\right| \leq \sqrt{6r_0}.
\end{align*}
Consequently, by definition of $\epsilon$ in Eq.~\eqref{eq:theorem-1-smallness-condition}, we have the following upper bound on $\epsilon$,
\begin{align*}
    \epsilon \leq \sqrt{6r_0}/(1 - 3r_0) \leq 0.25.
\end{align*}
Since $\epsilon < 1/3$, we are allowed to use theorem~\ref{thm:bound-on-error-growth-per-timestep}. Applying these bounds to terms in Eq.~\eqref{eq:single-qubit-teleportation-chain-iterative-bounds} gives a lower bound on $\left[\mathcal{P}_{(t,t-1)}\right]_{P,P}$,
\begin{align*}
    \left[\mathcal{P}_{(t,t-1)}\right]_{P,P} &\geq \left.\left(1- \frac{3\epsilon^3}{1 - 3\epsilon^2}\right)(1 - 3r_0)\right|_{\epsilon=\frac{\sqrt{6r_0}}{1 - 3r_0}} \\
    &\quad\quad - \left.\left(1 + \frac{3\epsilon^2}{1 - 3\epsilon^2}\right)\frac{(\sqrt{6r_0})^2}{1 - 3r_0}\right|_{\epsilon=\frac{\sqrt{6r_0}}{1 - 3r_0}} \\
    &\geq 1 - 17r_0\quad \mathrm{for}\ r_0 \leq 1 / 100.
\end{align*}
The last line is obtained via graphing. We can then repeatedly apply Eq.~\eqref{eq:theorem-1-relative-change-definition}, which gives
\begin{align*}
    \left[\overline{\mathcal{N}_t}\right]_{P,P} &\geq (1 - 17 r_0)^{t-1}\left[\overline{\mathcal{N}_1}\right]_{P,P}
\end{align*}
In addition, $\left[\overline{\mathcal{N}_1}\right]_{P,P} = \left[\mathcal{H}\mathcal{E}_1\mathcal{H}\right]_{P,P} \geq 1 - 3r_0 \geq 1 - 17r_0$, so we have
\begin{align*}
    \left[\overline{\mathcal{N}_t}\right]_{P,P}
    &\geq (1 - 17r_0)^t.
\end{align*}
Substituting into the formula for the average infidelity in Eq.~\eqref{eq:average-infidelity-definition} gives
\begin{align*}
   r(\overline{\mathcal{N}_t}) &\leq \frac{1}{2}[1 - (1 - 17r_0)^t] \leq \frac{17}{2}r_0t.
\end{align*}
Therefore, the average infidelity of the teleportation chain grows at most linearly with time.

\section{Circuit-level pure \texorpdfstring{$Z$}{Z}-coherent errors in MBEC map to physical Pauli errors}
\label{sec:circuitlvl-eithetaz-errors}

We now focus on an experimentally-motivated noise model~\cite{levine_parallel_2019, bluvstein_quantum_2022} which we call circuit-level errors with pure $Z$-coherence (theorem~\ref{thm:2}).
We analytically show that these errors completely decohere into Pauli errors in MBEC due to teleportation.  In particular, we show that these coherent errors map exactly onto Pauli channels on each physical qubit.
Using this exact and efficient mapping to a Pauli channel, one can use Pauli noise simulations to efficiently benchmark the logical performance of MBEC with these circuit-level coherent errors. The MBEC schemes that we analyze are based on foliation~\cite{bolt2016foliated}, as reviewed in section~\ref{sec:preliminaries}.

The main theorem that we will prove in this section is the following:

\begin{theorem}
    If the single-qubit error channel $\mathcal{N}_{\mathcal{L}}$ at each space-time location ${\mathcal{L}}$ in the foliation of a CSS code has pure $Z$-coherences, that is, it satisfies
    \begin{align}
        [\mathcal{N}_\mathcal{L}]_{P,XP}=[\mathcal{N}_\mathcal{L}]_{P,YP}=0,\ \forall P=\{I,X,Y,Z\}
    \end{align}
    and 
    \begin{align}
        [\mathcal{N}_\mathcal{L}]_{Z,Z}=1
        \label{eq:pure-z-condition}
    \end{align}
    then there exists a set of single-qubit Pauli channels $\mathcal{F}[\{\mathcal{N}_\mathcal{L}\}]$ that is equivalent to the original set of channels $\{\mathcal{N}_\mathcal{L}\}$ on the foliated code. These single-qubit Pauli channels $\mathcal{F}[\{\mathcal{N}_\mathcal{L}\}]$ are efficiently calculable from $\{\mathcal{N}_\mathcal{L}\}$, requiring only $W$ multiplications of $4\times 4$ matrices, where $W$ is upper bounded by the weight of the stabilizers and the number of $X$ or $Z$ stabilizers acting on any qubit from the CSS code.
    \label{thm:2}
\end{theorem}
\noindent Note that pure $Z$-coherent errors are $Z$-like coherent errors [as defined in Eq.~\eqref{eq:z-like-coherent-error}] with the additional constraint of Eq.~\eqref{eq:pure-z-condition}. As a consequence of this theorem, we have the following corollary:
\begin{corollary}
    Replacing $\mathcal{N}_\mathcal{L}$ with $\mathcal{F}[\{\mathcal{N}_\mathcal{L}\}]$ in the foliation circuit will lead to the same average logical channel after error correction.
    \label{cor:equivalent-ler}
\end{corollary}

We remark that we cannot directly use the results from the single-qubit teleportation chain in the MBEC setting due to two important complications. The first complication is that MBEC involves additional stabilizer measurements. For the single-qubit teleportation chain, we averaged over the random teleportation measurement outcomes to show that the resulting channel is twirled. In contrast, now we carefully choose teleportation and stabilizer measurements to show that the logical channel under physical Pauli noise is equivalent to the logical channel under physical coherent noise. The second complication is the presence of coherent errors on ancilla qubits used to measure stabilizers. These errors do not fit within the frameworks of existing numerical~\cite{bravyi2018correcting,marton2023coherent} and analytical~\cite{huang2019performance,beale2018quantum,iverson2020coherence} studies because a stabilizer measurement no longer collapses the state into a definite eigenspace of the stabilizers~\cite{beale2023randomized}. Despite these challenges, in Section~\ref{sec:coherent-errors-ancilla}, we show that we can still reduce coherent errors during stabilizer measurements to incoherent measurement errors. 

To explicitly show the reduction of circuit-level pure Z-coherent errors into Pauli errors in MBEC, we first manipulate the foliation circuit such that every coherent error on code qubits is preceded by an operator that would twirl (Eq.~\ref{eq:pauli-twirl}) it. We walk through these circuit manipulation steps in Section~\ref{sec:circuit-manipulation}. For the single-qubit teleportation chain, these circuit manipulations would be enough to show that coherent errors on code qubits can be replaced by their Pauli-twirled counterparts following a proof similar to the one presented in Section~\ref{sec:z-like-coherence}. However, to ensure corollary~\ref{cor:equivalent-ler} is satisfied, we find that there should be a slight reinterpretation of which measurement outcomes to keep constant while we average over other random measurement outcomes. After writing down explicit equivalent Pauli channels on code qubits, use the fact that errors are incoherent on code qubits to show that coherent errors on ancilla qubits can be modeled as incoherent measurement errors.

\begin{figure*}
    \centering
    \includegraphics[width=\linewidth]{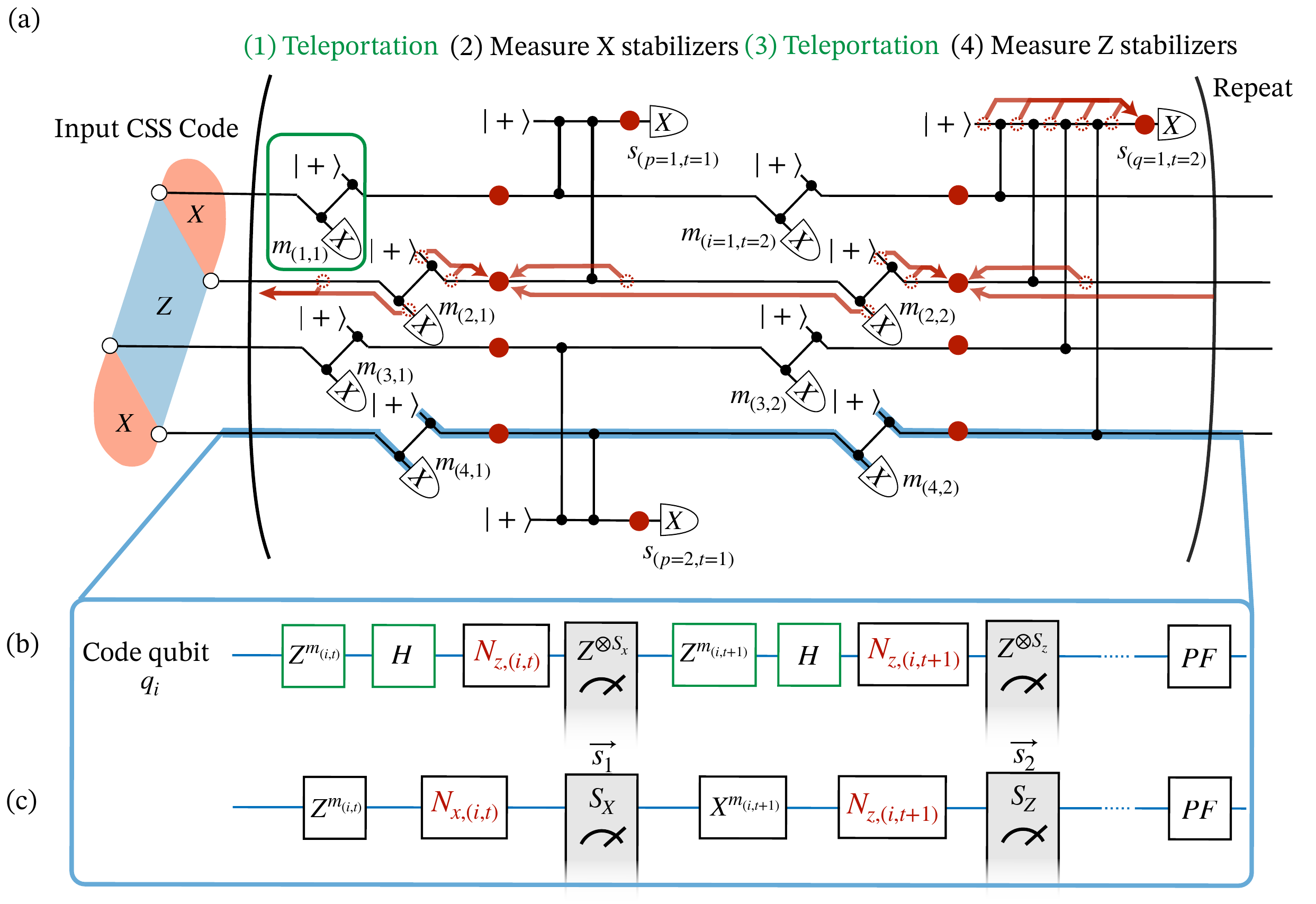}
    \caption{(a) We adopt an equivalent interpretation of the cluster state.  We start with qubits on the left in a code state of the $[[4,1,2]]$ CSS code. Then, we execute the circuit shown repeatedly to foliate this code into a cluster state. Each solid red circle in this circuit indicates an $N_{z,(\gamma,t)}=\alpha_{(\gamma,t)} I+\beta_{(\gamma,t)} Z$ over rotation error. These errors are the composition of multiple gate-level over rotation errors which have been commuted through the circuit (red arrows) from the noisy gate which they originated from (dotted, open circles). (b) We take the $i$th wire or code qubit $i$ from the foliation circuit in~(a) and construct an equivalent effective circuit in~(b) on this code qubit. We have replaced each teleportation circuit in~(a), circled in the green box, with the Kraus operator $H_iZ_i^{m_{(i,t)}}$, where $m_{(i,t)}$ is the random teleportation measurement outcome. Gray boxes indicate noisy stabilizer measurements. (c) Further manipulation of~(b) leads to the following operators on each code qubit.}
    \label{fig:Zerrors-clusterstate-final}
\end{figure*}

\subsection{Preliminaries}
\label{sec:preliminaries}
Consider a foliated $[[n,k,d]]$ CSS code. Foliation occurs via iterating through the following four sequential steps: (1) teleportation of each physical data qubit onto a fresh qubit, (2) an effective measurement of the $X$ stabilizers via ancilla qubits, (3) another teleportation of each data qubit onto fresh qubit, and (4) measurement of the $Z$ stabilizers via ancilla qubits. Thus, each data qubit of the $[[n,k,d]]$ code moves through a teleportation chain Fig.~\ref{fig:single-qubit-teleportation-chain}~(b) while simultaneously the ancilla qubits perform stabilizer measurements~\cite{bolt2016foliated}. Recall from Fig.~\ref{fig:single-qubit-teleportation-chain} that every teleportation applies a Hadamard; therefore, the measurement of an $X$-like stabilizer $X^{\otimes S_x}$, which is carried out after odd teleportation steps, is reduced to the measurement of a $Z$-like stabilizer $Z^{\otimes S_x}$ supported on the same qubits. Thus, each stabilizer can be measured by performing CZ gates with an ancilla prepared in the $\ket{+}$ state followed by an $X$-basis measurement of the ancilla. The four steps of foliation are repeated $L$ times to perform $L$ rounds of syndrome extraction. As an example, Fig.~\ref{fig:Zerrors-clusterstate-final} illustrates foliation with the [[4,1,2]] code.

For the rest of Section~\ref{sec:circuitlvl-eithetaz-errors}, we refer to each single-qubit teleportation chain as an effective \textit{code qubit}. We index these code qubits $q_i$ by $i\in[1..n]$. Each full round of syndrome extraction teleports the data qubit twice, therefore, the Pauli frame of each code qubit $q_i$ changes $2L$ times, according to the measurement outcomes of each teleportation circuit, $m_{(i,t)} \in \{0,1\}$ where $t\in [1..2L]$. Additionally, foliation produces measurement results $s_{p,t}$ and $s_{q,t}$ of the $X-$ and $Z-$ stabilizers of the original CSS code, where $p\in[1..N_{S_X}]$ and $q\in[1..N_{S_Z}]$. Here, $N_{S_X}$ and $N_{S_Z}$ are the number of $X$ and $Z$ stabilizers of the CSS code respectively. By construction of the foliation, the alternating $X-$ and $Z-$ stabilizer measurements mean that $s_{p,t}$ only exists for odd $t$, while $s_{q,t}$ only exists for even $t$. To return to the initial Pauli frame at the end of the foliation, we keep track of the outcomes $m_{(i,t)}$ in software, and apply a Pauli frame update $\pf{2L}$, defined in Sec.~\ref{sec:single-qubit-teleportation-chain}, Eq.~\eqref{eq:single-qubit-teleportation-pauli-frame}, on each code qubit at the end of the foliation. 

To perform error correction on a foliated code, we need not only stabilizer measurement outcomes $\vec{s}_{t}=(s_{p,t})$ or $(s_{q,t})$ but also teleportation measurement outcomes $\{m_{(i,t)}\}$. To illustrate this, we show the effective evolution of a code qubit $q_i$ after three rounds of foliation.
\begin{widetext}
    \centering
    \includegraphics[width=0.6\textwidth]{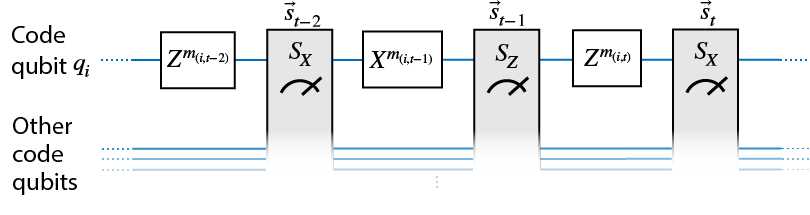}
\end{widetext}
Recall from Eq.~\eqref{eq:single-qubit-teleportation-pauli-frame:b} that a random $X^{m_{(i,t)}}$ arises due to the application of a Hadamard from each teleportation. From this evolution, we see that even in the absence of errors, subsequent $\vec{s}_t$ measurement outcomes will flip depending on induced random Paulis $Z^{m_{(i,t)}}$ and $X^{m_{(i,t-1)}}$ during teleportation. 
This dependence of $\{\vec{s}_{t}+\vec{s}_{t-2}\}$ on $\{m_{(i,t)}\}$ is removed by defining standard cluster state stabilizer measurement outcomes, $\{\vec{s}_t + \vec{s}_{t-2}+\vec{u_t}(m_{(i,t)})\}$. Here, $\vec{u_t}$ is defined as the sum of $m_{(i,t)}$ on all qubits in the support of the stabilizer (App.~\ref{app:cluster-state-stabilizers}). All of these cluster state stabilizers are $0$ in the absence of errors and are used for decoding and error correction.

Now that we have discussed the preliminaries of foliation, the aim of the next section~\ref{sec:circuit-manipulation} is to manipulate a foliation circuit with errors of pure $Z$-coherence at every space-time location.

\subsection{Circuit manipulation with pure \texorpdfstring{$Z$}{Z}-coherent errors} 
\label{sec:circuit-manipulation}

In the circuit-level noise model under consideration, single-qubit pure $Z$-coherent errors $\mathcal{N}_{z,(\gamma,t,w)}$, (dotted red circles in Fig.~\ref{fig:Zerrors-clusterstate-final}~(a)) are applied at every space-time location $(\gamma,t,w)$ in the circuit. These spacetime locations occur between every operation including state preparation, two-qubit gates, and measurement. It is convenient to label each space-time location with $(\gamma,t,w)$, where $\gamma\in[1..n+N_{S_X}+N_{S_Z}]$ indicates the code qubit or ancilla qubit suffering the error, and $t\in[1..2L]$ indicates the round of foliation. $\gamma$ is an index that goes over all code qubits and ancilla qubits, where indices $[1,n]$ are designated for code qubits, indices $[n+1,n+N_{S_X}]$ are ancilla qubits measuring $X$ stabilizers, and indices $[n+N_{S_X}+1, n+N_{S_X}+N_{S_Z}]$ are ancilla qubits measuring $Z$ stabilizers. Here we define a round of foliation such that it starts and ends, respectively, with the initialization and measurement of the physical qubit in the teleportation chain.
Finally, $w\in[1..W_{\gamma}]$ enumerates spacetime locations where errors can happen in a single foliation round. In general, $W_{\gamma}$ is bounded by the weight of the stabilizers of the code or the number of $X$ or $Z$ stabilizers measured on each data qubit. 
To illustrate, consider the example in Fig.~\ref{fig:Zerrors-clusterstate-final}~(a). The ancilla qubit measuring the $Z$ stabilizer on step $(4)$ has $W_\gamma=5$.

With indices $\gamma,t$ and $w$, we can describe the location of all pure $Z$-coherent errors, which in general, may be inhomogeneous over $(\gamma,t,w)$. Although $\mathcal{N}_{z,(\gamma,t,w)}$ may, in general, have Kraus rank $>1$, for simplicity, we will show the proof for $\mathcal{N}_{z,(\gamma,t,w)}=N_{z,(\gamma,t,w)}\cdot N_{z,(\gamma,t,w)}^\dagger$ where $N_{z,(\gamma,t,w)}=\alpha_{(\gamma,t,w)} I+\beta_{(\gamma,t,w)} Z$.

Within each foliation round, we choose to commute all single-qubit errors on the code qubits to the space-time location before the stabilizer measurement of the same round. These commutations are easily possible because all the noise channels commute with the gates. On ancilla qubits, we commute all single-qubit errors to right before the measurement. In the example of Fig~\ref{fig:Zerrors-clusterstate-final}~(a), these commutations are represented by the red arrows, where we have commuted errors to the locations in the circuit with closed, red circles. As a result, all single qubit errors at each physical qubit $\gamma$ and in each foliation round $t$ are coherently combined into
\begin{align}
    \mathcal{N}_{z,(\gamma,t)}=\prod_{w=1}^{W_{\gamma}}\mathcal{N}_{z,(\gamma,t,w)}=N_{z,(\gamma,t,w)}\cdot N_{z,(\gamma,t,w)}^\dagger\\
    N_{z,(\gamma,t,w)}=\alpha_{(\gamma,t)}I + \beta_{(\gamma,t)}Z
    \label{eq:spacetime-coherent-addition}
\end{align}
For the rest of the section, when a noise channel $\mathcal{N}_z$ is indexed by two variables, it refers to a channel where errors have already been partially coherently combined in this manner.
Constructive interference due to the composition of these channels over $W_{\gamma}$ is efficiently calculable-- we simply need to compose $W_{\gamma}$ single-qubit channels. This is the only place in the proof where the composition of coherent channels is necessary.
\subsection{Coherent errors on code qubits}
\label{sec:coherent-errors-on-code-qubits}

In this section, we show the decoherence of coherent errors on only the code qubits, which recall from Section~\ref{sec:preliminaries}, are indexed by $i$ instead of $\gamma$.
After circuit manipulation as described in the previous section, the evolution of each code qubit for two consecutive rounds of foliation can be described by the sequence in Fig.~\ref{fig:Zerrors-clusterstate-final}~(b). The $Z^{m_{(i,t)}}$ followed by $H$ in Fig.~\ref{fig:Zerrors-clusterstate-final}~(b) arises due to the effective operations of the noiseless teleportation circuit from Eq.~\eqref{eq:kraus-teleportation}, where $m_{(i,t)}$ is a uniform, random bit. By commuting Hadamard operators from odd time steps $t$ to Hadamard operators at even $t$, we arrive at effective operators in Fig.~\ref{fig:Zerrors-clusterstate-final}~(c). Here, half of the pure $Z$-coherent errors transform into pure $X$-coherent errors of the form $N_{x,(i,t)} = \alpha_{(i,t)} I+ \beta_{(i,t)} X,\ \forall t\in \mathrm{odd}$. We have changed indices to $i$ to indicate that we refer to lumped coherent errors in Eq.~\ref{eq:spacetime-coherent-addition} on code qubits $i\in[1,n]$. Further, the $\mathrm{CZ}$ gates used to measure $Z^{\otimes S_x}$ are transformed by the Hadamards into $\mathrm{CX}$ gates, thereby measuring the $X^{\otimes S_x}$ operator, by the design of the foliation procedure. Finally, the random $Z$ operators on even $t$ are transformed into random $X$ operators. 

Now, we annihilate the random teleportation-induced Paulis from Fig.~\ref{fig:Zerrors-clusterstate-final}~(c), $Z_i^{m_{(i,t)}}$ and $X_i^{m_{(i,t+1)}}$, with the Pauli frame update, $\pf{2L}$:
\begin{widetext}
\centering
\includegraphics[width=0.7\textwidth]{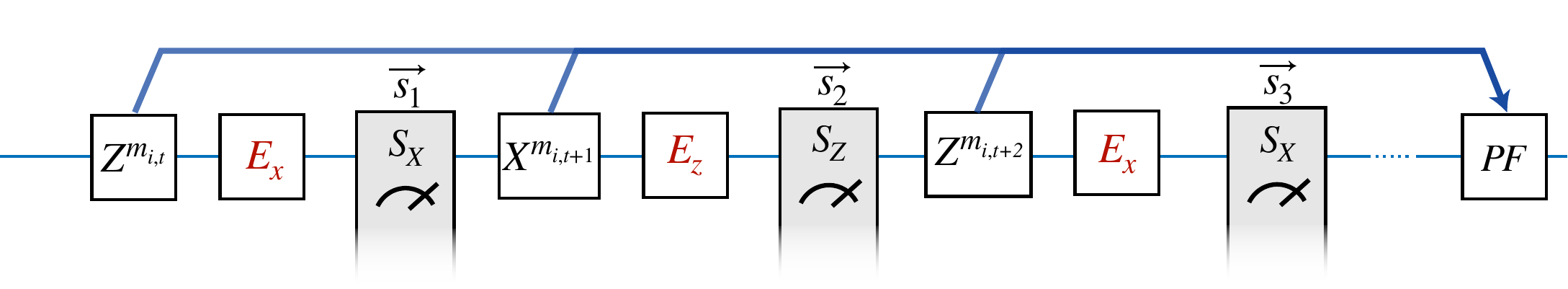}
\end{widetext}
As a result, the measurement outcomes $\vec{s}_{t}$ will be flipped by $m_{(i,t)}$. Intuitively, one can understand these measurement flips as a reinterpretation of each stabilizer measurement outcome $\vec{s}_{t}$ depending on the Pauli frame induced by all preceding random teleportation measurements. These set of flipped measurement outcomes exactly generates the cluster state stabilizer outcomes as discussed in section~\ref{sec:preliminaries}. Additionally, the movement of these random teleportation Paulis to the end of the circuit flips the sign of $\beta_{(i,t)}$ in $N_{x,(i,t)}$ or $N_{z,(i,t)}$. This is equivalent to conjugating each $\mathcal{N}_{x,(i,t)}$ or $\mathcal{N}_{z,(i,t)}$ with the Pauli frame $\pf{t}$. We denote this conjugated channel as $\mathcal{N}_{x,(i,t)}'$ or $\mathcal{N}_{z,(i,t)}'$ which has the following form,
\begin{align}
    \mathcal{N}_{x,(i,t)} \rightarrow \mathcal{N}_{x,(i,t)}' &= \pf{t}^{-1}\mathcal{H}\mathcal{N}_{x,(i,t)}\mathcal{H}\pf{t} &t\ \mathrm{odd} \notag\\
    &= \mathcal{E}_{I,(i,t)}+(-1)^{z_{(i,t)}}\mathcal{E}_{X,(i,t)},
    \label{eq:conjugated-x-channel} \\
    \mathcal{N}_{z,(i,t)} \rightarrow \mathcal{N}_{z,(i,t)}' &= \pf{t}^{-1}\mathcal{N}_{z,(i,t)}\pf{t} &t\ \mathrm{even} \notag\\
    &= \mathcal{E}_{I,(i,t)}+(-1)^{x_{(i,t)}}\mathcal{E}_{Z,(i,t)},
    \label{eq:conjugated-z-channel}
\end{align}
where $x_{(i,t)}$ and $z_{(i,t)}$ are defined as
\begin{align}
    x_{(i,t)}&=m_{(i,2)}+m_{(i,4)}+...+m_{(i,t)} & t\ \mathrm{even}, \\
    z_{(i,t)}&=m_{(i,1)}+m_{(i,3)}+...+m_{(i,t)} & t\ \mathrm{odd},
\end{align}
Additionally, $\mathcal{E}_{I,(i,t)}$, $\mathcal{E}_{X,(i,t)}$ and $\mathcal{E}_{Z,(i,t)}$ are defined as
\begin{align}
    \mathcal{E}_{I,(i,t)}(\cdot) &=
    \begin{cases}
        |\alpha_{(i,t)}|^2 I \cdot I + |\beta{(i,t)}|^2 X \cdot X,&\forall t\in\mathrm{odd}\\
        |\alpha_{(i,t)}|^2 I \cdot I + |\beta{(i,t)}|^2 Z \cdot Z,&\forall t\in\mathrm{even}\\
    \end{cases}
\end{align}
\begin{align}
    \mathcal{E}_{X,(i,t)}(\cdot) &= \alpha_{(i,t)}\beta_{(i,t)}^* I \cdot X + \beta_{(i,t)} \alpha_{(i,t)}^* X \cdot I,& t\ \mathrm{odd} \\
    \mathcal{E}_{Z,(i,t)}(\cdot) &= \alpha_{(i,t)}\beta_{(i,t)}^* I \cdot Z + \beta_{(i,t)} \alpha_{(i,t)}^* Z \cdot I,& t\ \mathrm{even}
\end{align}
Suggestively, we have decomposed the noise channel $\mathcal{N}'_{x,(i,t)}$ and $\mathcal{N}'_{z,(i,t)}$ 
in Eqs.~\eqref{eq:conjugated-x-channel}~and~\eqref{eq:conjugated-z-channel} to group terms of the channel that are independent of the random variables 
$x_{(i,t)}$ and $z_{(i,t)}$, and terms whose sign depends on $x_{(i,t)}$ and $z_{(i,t)}$.  Concretely, $\mathcal{E}_{I,(i,t)}$ is independent of $x_{(i,t)}$ and $z_{(i,t)}$ and will remain invariant after averaging over these random variables. Meanwhile, $\mathcal{E}_{X,(i,t)}$ and $\mathcal{E}_{Z,(i,t)}$ will disappear after averaging over the $-1$ sign dependent on random $x_{(i,t)}$ or $z_{(i,t)}$ respectively.

Explicitly, after this commutation, we arrive at the following evolution on all code qubits,
\begin{widetext}
    \centering \includegraphics[width=0.8\textwidth]{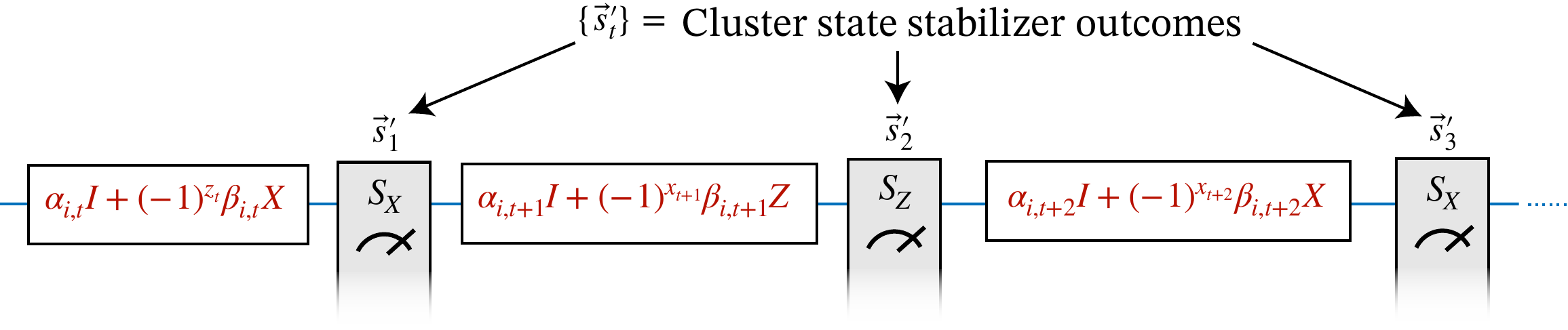}
\end{widetext}
To show the decoherence of errors, we fix the cluster state stabilizers $\{\vec{s}_t'\}$, and average over independent and random $\{m_{(i,t)}\}$. We remark that each conjugated error $\mathcal{N}_{x,(i,t)}'$ and $\mathcal{N}_{z,(i,t)}'$ is uncorrelated with each other. As a result, the effective channel can be written as the following,
\begin{align}
    \overline{\mathcal{N}_{\mathrm{tot}}}=\!\!\!\!\prod_{t=\{1,3,...2L-1\}}\!\!\!\!\mathcal{M}[\vec{s}_{t+1}']\overline{\mathcal{N}_{z,(i,t+1)}'} \mathcal{M}[\vec{s}_t']\overline{\mathcal{N}_{x,(i,t)}'},
    \label{eq:averaged-code-qubit-channel}
\end{align}
where $\mathcal{M}[\vec{s}_t']$ is the channel associated with the stabilizer measurements given measurement outcome $\vec{s}_t'$. 
As in Section~\ref{sec:z-like-coherence}, the averaged channels $\overline{\mathcal{N}_{x,(i,t)}'}$ and $\overline{\mathcal{N}_{z,(i,t+1)}'}$ are simply the twirled versions of the unconjugated error channels $\mathcal{N}_{x,(i,t)}$ and $\mathcal{N}_{z,(i,t)}$ respectively. In other words, coherent errors $\mathcal{N}_{x/z,(i,t)}$ at each space-time location on the code qubits can replaced with an equivalent Pauli channel
\begin{align}
    \overline{\mathcal{N}_{x/z,(i,t)}'}=
    \begin{cases}
        \mathcal{P}_x\!\left[p = |\beta_{(i,t)}|^2\right] & \text{if $t\ \mathrm{odd}$}\\
        \mathcal{P}_z\!\left[p = |\beta_{(i,t)}|^2\right] & \text{if $t\ \mathrm{even}$}
    \end{cases}
\end{align}
and produce the same logical action. 
Here, $\mathcal{P}_x\big[p] = (1-p) \mathcal{I} + p \mathcal{X}$ is a single-qubit bit-flip channel with bit-flip probability $p$. $\mathcal{P}_z\big[p] = (1-p) \mathcal{I} + p \mathcal{Z}$ is a single-qubit phase-flip channel with phase-flip probability $p$.

Recall that each $\mathcal{N}_{z,(i,t)}$ is actually the coherent addition of $W_i$ single-qubit channels from Eq.~\eqref{eq:spacetime-coherent-addition}. Therefore, the channel at each space-time location on code qubits $(i,t,w)$ can be replaced by Pauli channels such that $W_i$ compositions of these channels gives a cumulative bit or phase-flip rate of $|\beta_{(i,t)}|^2$. That is, the coherent channels $\{\mathcal{N}_{(i,t,w)}\}$ at each $(i,t,w)$ are replaced by Pauli channels,

\begin{align}
    &\mathcal{N}_{(i,t,w)}\mapsto \nonumber \\
    &\begin{cases}
        \mathcal{P}_x\!\left[p = \frac{1}{2}\left(1-\sqrt[W_i]{1-2|\beta_{(i,t)}|^2}\right)\right],& \text{$t\ \mathrm{odd}$},\ \forall w\\
        \mathcal{P}_z\!\left[p =\frac{1}{2}\left(1-\sqrt[W_i]{1-2|\beta_{(i,t)}|^2}\right)\right],& \text{$t\ \mathrm{even}$},\ \forall w\\
    \end{cases}
    \label{eq:code-qubit-pauli-channel}
\end{align}
Therefore, we have shown that after averaging, all coherent errors on code qubits are equivalent to Pauli errors for a given set of stabilizer measurement outcomes $\{\vec{s}_t'\}$.

\subsection{Coherent errors on ancilla qubits} \label{sec:coherent-errors-ancilla}

In this section, we will show that coherent errors on the ancilla qubits mediating stabilizer measurements are effectively incoherent measurement errors. Combined with the conclusions of the previous section~\ref{sec:circuit-manipulation}, this means that all circuit-level errors with $Z$ coherence in the foliated code are incoherent. We will then tie our results from the previous section and this one to prove corollary~\ref{cor:equivalent-ler} of theorem~\ref{thm:2}, which says that we can simulate the logical channel of MBEC under coherent errors with Pauli errors and still reproduce the same logical channel.

To show that coherent errors on ancilla qubits during foliation are Pauli, we first define a perfect stabilizer measurement, and then consider changes to this measurement once pure $Z$-coherent errors are added. Without loss of generality, we will examine one stabilizer $S$, and the same procedure can be applied to each stabilizer measurement in foliation. 

Suppose we measure a stabilizer generator $S$ and obtain measurement outcome $s$.

Then, in the absence of errors, this applies the following perfect projector on the data qubits,
\begin{align}
    \Pi[s] = \frac{1}{2} [I+(-1)^{s}S].
    \label{eqn.perfect-projector}
\end{align}
Now, let us consider the circuit-level noise model with pure $Z$-coherent errors at every space time location $(\gamma,t,w)$. Recall in Fig.~\ref{fig:Zerrors-clusterstate-final}~(a), we pushed the errors to the spacetime location right before the measurement of the ancilla qubits. Consequently it is natural to consider an error $N_{z}=\alpha I+\beta Z$ on the ancilla qubit measuring the stabilizer $S$. This changes the projector $\Pi[s]$ to an imperfect projector $\tilde\Pi[s]$ of the following form,

\begin{align}
    \tilde{\Pi}[s] &=  \alpha \Pi[s]
    +\beta\Pi[1-s]
    \label{eqn:imperfect-projection}
\end{align}

Until now, we have not made any assumptions about the code qubit state that $\tilde{\Pi}[s]$ acts on. However, in section~\ref{sec:circuit-manipulation}, we showed that the effective error channel acting on the code qubits was Pauli. Therefore, we can make use of the following lemma,
\begin{lemma}
    A Pauli error channel $\mathcal{P}$ on a code state $\rho_\mathrm{c}$ results in a mixed state of definite stabilizer eigenvalues $\vec{\delta}\in\{0,1\}^{n-k}$.
\end{lemma}
\noindent Given the effective noise channel on the code qubits is Pauli from Section~\ref{sec:coherent-errors-on-code-qubits}, this means that the state of the code qubits is always a definite stabilizer eigenstate, rather than a superposition of stabilizer eigenstates.  
The next key observation is that when imperfect projectors from Eq.~\ref{eqn:imperfect-projection} act on a definite stabilizer eigenstate $\ket{\vec{\delta}}$, we effectively get measurement errors with the same statistics as the twirled version of $\mathcal{N}_z$. 
To see this, we perform a POVM on a specific stabilizer eigenstate $\rho_0=\ket{\vec{\delta}}\bra{\vec{\delta}}$ and argue that by linearity, the following steps will hold for code qubits in mixed states $\sum_{i}c_i\ket{\vec{\delta_i}}\bra{\vec{\delta_i}}$. Let the $l$th index of $\vec{\delta}$ denote the eigenstate of stabilizer $S$. Then, after noisy stabilizer measurement,

\begin{align}
    &\tilde{\Pi}[\delta_l]\ \rho_0\ \tilde{\Pi}^\dagger[\delta_l] + \tilde{\Pi}[1-\delta_l]\ \rho_0\ \tilde{\Pi}^\dagger[1-\delta_l] \notag \\ 
    &=\bigg[|\alpha|^2\Pi[\delta_l]\rho_0\Pi^\dagger[\delta_l]\bigg]_{\!\delta_l}\!\!\!\!+\bigg[|\beta|^2\Pi[\delta_l] \rho_0 \Pi^\dagger[\delta_l]\bigg]_{\!1-\delta_l}\!\!\!\!\!\!\propto \rho_0
\end{align}
In the second line of equations, we seen that all terms except for those $\propto\Pi[\delta_l]\rho_0\Pi^\dagger[\delta_l]$ are canceled out. The first remaining term is from $\tilde{\Pi}[\delta_l]$ projector (see the bracket subscript $\delta_l$) where the measurement outcome is $s=\delta_l$. The second remaining term is from the $\tilde{\Pi}[1-\delta_l]$ projector. 
From the coefficients of these terms in the density operator, we see that we obtain the correct outcome, $s=\delta_l$, with probability $|\alpha|^2$ and incorrect outcome, $s=1-\delta_l$, with probability $|\beta|^2$. Finally, we observe that the final state is still $\propto \rho_0$, introducing no new coherent errors onto the state of the code qubits. Therefore, given a coherent error $N_{z}$ on the ancilla qubit before measurement, the effective action on the code qubits initially in $\ket{\vec{\delta}}$ can be written as

\begin{align}
    \tilde{\Pi}[s]\ket{\vec{\delta}}
    \propto\Pi[f(s)]\ket{\vec{\delta}}
    \label{eq:reduce-imp-proj}
\end{align}

where $f:\{0,1\}\rightarrow\{0,1\}$ is a binary symmetric channel defined as:
\begin{align}
    \mathrm{Pr}[f(s\neq \delta_l)]=|\beta|^2
    \label{eq:classical-measurement-err-prob}
\end{align}

All previous arguments are directly extendable to measurements of multiple stabilizers $S_{(l,t)}$. Let each ancilla qubit be indexed by $l$ with pure Z-coherent errors $\mathcal{N}_{z,(l,t)}$ (Eq.~\ref{eq:spacetime-coherent-addition}) at space time locations right before measurement. Then, these errors reduce to classical measurement errors with probability Eq.~\eqref{eq:classical-measurement-err-prob}. Equivalently, this means that coherent errors $\mathcal{N}_{z,(l,t,w)}$ on ancilla qubits can be replaced by the following Pauli channel
\begin{align}
    \mathcal{P}_z\!\left[p = \frac{1}{2}\left(1-\sqrt[W]{1-2|\beta_{(l,t)}|^2}]\right)\right]
    \label{eq:ancilla-qubit-pauli-error}
\end{align}
to achieve the same effective measurement error and projection operator on the code qubits.

Sections~\ref{sec:coherent-errors-on-code-qubits} and~\ref{sec:coherent-errors-ancilla} show that both code qubit and ancilla coherent errors during foliation equivalently reduce to incoherent errors. In other words, all pure $Z$-coherent errors in the construction and measurement of the cluster state can be captured with purely Pauli errors. In the next section, we remark on how these results imply corollary~\ref{cor:equivalent-ler} of theorem~\ref{thm:2}.

\subsubsection{Pauli channels yield equivalent logical channels as coherent errors}
We now justify corollary~\ref{cor:equivalent-ler} of theorem~\ref{thm:2}. Using results from sections~\ref{sec:circuit-manipulation}~and~\ref{sec:coherent-errors-ancilla}, this corollary states that replacing error channels $\mathcal{N}_{z,(\gamma,t,w)}$ with Eq.~\eqref{eq:code-qubit-pauli-channel} on code qubits and Eq.~\eqref{eq:ancilla-qubit-pauli-error} on ancilla qubits gives the same logical channel as if there were physical coherent errors at every space-time location of the circuit.

For corollary~\ref{cor:equivalent-ler} to hold, we require that the logical channel after foliation and error correction assuming physical coherent errors must be equal to the same foliation and error correction circuit  with physical Pauli errors. That is,
\begin{align}
    \mathcal{E}_\mathrm{L}[\mathcal{N}_{z,(\gamma,t,w)}]=\mathcal{E}_\mathrm{L}[\mathcal{F}[\mathcal{N}_{z,(\gamma,t,w)}]]
    \label{eq:logical-channel-equivalence}
\end{align}

The decoder and recovery operator of $\mathcal{E}_\mathrm{L}[\mathcal{N}_{z,(\gamma,t,w)}]$ only  takes in stabilizer measurement outcomes and corrects based on those measurement outcomes. If the coherent and Pauli channels are equivalent given these set of syndrome outcomes, then the logical channels are equivalent. Since we have shown the equivalence of a physical coherent channel with a Pauli channel given syndrome outcomes $\vec{s}_t'\forall t$, then we are guaranteed to have the same logical channel. As a consequence, we can guarantee that the logical channel produced from a classically-efficient Clifford simulator on noise $\mathcal{P}[\mathcal{N}_\mathcal{L}]$ can be used to benchmark the exact performance of circuit-level coherent errors in MBEC.

In conclusion, we have analyzed circuit-level pure $Z$-coherent errors and have shown that all physical errors $N_{z,(\gamma,t,g)}$ are equivalent to Pauli error channels on code qubits (Eq.~\eqref{eq:code-qubit-pauli-channel}) and ancilla qubits (Eq.~\eqref{eq:ancilla-qubit-pauli-error}). These reductions to incoherent errors are exact, showing that coherent errors intrinsically become incoherent on foliated CSS codes due to the random Pauli frames induced by teleportation.

\section{Conclusion and discussion}
\label{sec:discussion}
We have shown that teleportation in MBEC effectively decoheres coherent errors into incoherent Pauli errors. In particular, we showed that the average infidelity of single-qubit teleportation grows at worst linearly with the number of teleportations, and that circuit-level single-qubit pure $Z$-coherent errors map exactly into a physical Pauli error model in MBEC. This correspondence between coherent and incoherent error models allows us to perform efficient simulation as well as repurpose existing analytical results on incoherent Pauli errors. 

One important consequence of our results is that we can obtain an analytical lower bound on the error threshold for circuit-level $e^{i\theta Z}$ errors in the teleported surface code, i.e., the RHG cluster state~\cite{raussendorf_one-way_2001, raussendorf_long-range_2005, raussendorf_fault-tolerant_2007, raussendorf_fault-tolerant_2006}. To see this, we first coherently combine the $e^{i\theta Z}$ errors at all 5 spacetime locations on each physical qubit following section~\ref{sec:circuit-manipulation}. Afterwards, the error model on the teleported surface code becomes an $e^{i5\theta Z}$ error after the teleportation of each code qubit after each teleportation as well as before the measurement of each ancilla qubit. Our result then shows that teleportation decoheres these combined errors into physical Pauli $Z$ errors with probability $p = \sin^2(5\theta)$. The resulting Pauli $Z$ error model on the RHG cluster state is equivalent to the circuit-based implementation of the surface code with independent data qubit idle $X$ and $Z$ error rate $p_X=p_Z=p$ and measurement error rate $p_M=p$. Now we can apply known threshold theorems (for example \cite{DennisTopoQuantumMemory, fowler_proof_2012-1}) for the surface code under Pauli errors. Adapting the proof in~\cite{fowler_proof_2012-1}, we can show that the threshold on $p$ is lower bounded by $p_{\mathrm{th}} \geq 1/[2 (B-1)]^2$, where $B$ is the maximum number of nearest neighbors in the decoding graph. In our case, we set $B=6$ for the cubic lattice decoding graph and obtain $p_{\mathrm{th}} \geq 1/100$. Therefore, the threshold angle $\theta_{\mathrm{th}}$ for circuit-level $e^{i\theta Z}$ errors in the teleported surface code satisfies the lower bound $\theta_{\mathrm{th}} = \arcsin(\sqrt{p_{\mathrm{th}}})/5 \geq \arcsin(1/10)/5 \approx 0.02$. A more accurate estimate of $\theta_{\mathrm{th}} \approx0.035$ can be computed from the numerically simulated value $p_{\mathrm{th}} \approx 0.030$ (see Supplementary Figure 5 in~\cite{claes_tailored_2023} or Figure 17 in~\cite{yoder_surface_2017}).

Additionally, our result implies that the logical error channel after error correction is a Pauli error channel. Therefore, we find that logical errors cannot constructively interfere over many rounds of error correction~\cite{beale2018quantum,huang2019performance} in MBEC.

Our result has attractive implications for superconducting and neutral atom platforms. Currently, coherent errors on this platform are mitigated via dynamical decoupling and echo pulses~\cite{acharya_quantum_2025, bluvstein_logical_2024}. Performing MBEC may obviate the need for these error-canceling pulses as long as the coherent errors are not too large. Additionally, because selective single-qubit gates \cite{graham_multi-qubit_2022} can sometimes be noisier and harder to do than global single-qubit gates, it may be advantageous to exploit the intrinsic capability of MBEC to decohere residual coherent errors instead of performing randomized compiling.

While we have proven the decoherence of coherent errors for single-qubit pure $Z$-coherent errors, we have not generalized our proof to other errors, such as arbitrary-axis single-qubit rotations or $e^{i\theta ZZ}$ errors that may accompany CZ gates. In such cases, it may be necessary to perform numerical simulations where analytics fail. We leave the generalization of our result to other errors and the simulation of MBEC under coherent errors to future work. Finally, it may be possible to generalize our result to other protocols that make use of teleportation or frequent random measurements, such as Floquet codes~\cite{hastings_dynamically_2021}, teleportation-based logical gates~\cite{berthusen_toward_2025}, and schemes that use teleportation to create long-range connections using only geometrically-local connectivity~\cite{choe_how_2024}.

\section{Acknowledgments}
We thank Ken Brown, Tyler LeBlond, and Peter Groszkowski for useful discussions. This work was supported by  DARPA MeasQuIT (HR00112490363) and U.S. Army Research Office (W911NF-23-1-0051).
\newpage

\bibliography{apssamp}

\clearpage
\appendix

\section{Proof of lower and upper bounds}
\label{app:bounds-proof}
Now we provide a proof for theorem~\ref{thm:bound-on-error-growth-per-timestep}. The main idea of the proof is to rewrite $\overline{\mathcal{N}_t}$ as some linear map $\mathcal{P}_{(t, t-1)}$ composed with $\overline{\mathcal{N}_{t-1}}$, and then find an approximation for $\mathcal{P}_{(t, t-1)}$. The map $\mathcal{P}_{(t, t-1)}$ will tell us how much the diagonal elements $\left[\overline{\mathcal{N}_t}\right]_{P,P}$ change from $t-1$ to $t$. Our proof is organized into a sequence of lemmas. Instead of directly using the definition of the unconditional average channel $\overline{\mathcal{N}_t}$ in Eq.~\eqref{eq:single-qubit-teleportation-average-error-channel}, we first derive an exact iterative formula that relates the conditional average of $\mathcal{N}_t$ at time $t$ to the that at time $t-1$ (lemma~\ref{lem:iterative-expressions}). From this formula, we extract a relation between $\overline{\mathcal{N}_t}$ and $\overline{\mathcal{N}_{t-1}}$, which confirms that the leading order term in $\mathcal{P}_{(t, t-1)}$ is the Pauli twirl of the physical error channel, as well as provides an exact correction term (lemma~\ref{lem:leading-approximation}). We then informally refine the expression for $\mathcal{P}_{(t, t-1)}$ to include a correction term coming from interference between neighboring timesteps (lemma~\ref{lem:next-approximation}). Finally, we bound the errors in the informal expression, which concludes the proof of theorem~\ref{thm:bound-on-error-growth-per-timestep}.

\begin{lemma}
    \label{lem:iterative-expressions}
    Let $\overline{\mathcal{N}_t|\pf{t}} $ denote the average of $\mathcal{N}_t$ conditioned on the most recent Pauli frame being $\pf{t}$,
    \begin{align}
        \overline{\mathcal{N}_t|\pf{t}} &= \!\!\!\!\!\sum_{\pf{t-1}, \cdots, \pf{1}}
        \!\!\!\!\!\mathcal{E}_t'\cdots \mathcal{E}_1'
        \Pr(\pf{t-1}, \cdots, \pf{1}|\pf{t}). \label{eq:single-qubit-teleportation-chain-conditional-average-channel}
    \end{align}
    Then for all $t > 1$ we have
    \begin{align}
        \overline{\mathcal{N}_t|\pf{t}} &= \sum_{\pf{t-1}}\Pr(\pf{t-1}| \pf{t})\ \mathcal{E}_t'\ \overline{\mathcal{N}_{t-1}|\pf{t-1}},
        \label{eq:single-qubit-teleportation-chain-conditional-average-channel-iterative}
    \end{align}
    and for all $t$ we have
    \begin{align}
        \overline{\mathcal{N}_t} &= \sum_{\pf{t}}\Pr(\pf{t})\ \overline{\mathcal{N}_t|\pf{t}}. \label{eq:single-qubit-teleportation-chain-average-conditional-channel}
    \end{align}
\end{lemma}
\begin{proof}[Proof of lemma~\ref{lem:iterative-expressions}]
    Eq.~\eqref{eq:single-qubit-teleportation-chain-average-conditional-channel} follows from the definition of conditional probabilities. To show Eq.~\eqref{eq:single-qubit-teleportation-chain-conditional-average-channel-iterative} from Eq.~\eqref{eq:single-qubit-teleportation-chain-conditional-average-channel}, we notice that by definition of conditional probabilities,
    \begin{align*}
        &\Pr(\pf{t-1}, \cdots, \pf{1}|\pf{t}) \nonumber \\
        &= \Pr(\pf{t-1}|\pf{t})\ \Pr(\pf{t-2}, \cdots, \pf{1}|\pf{t}, \pf{t-1}) \\
        &=\Pr(\pf{t-1}|\pf{t})\ \frac{\Pr(\pf{t}, \cdots, \pf{1})}{\Pr(\pf{t}, \pf{t-1})} \\
        &=\Pr(\pf{t-1}|\pf{t}) \nonumber \\
        &\quad\times \frac{\Pr(\pf{t}|\pf{t-1}, \cdots\!, \pf{1})\Pr(\pf{t-1}, \cdots, \pf{1})}{\Pr(\pf{t}|\pf{t-1})\Pr(\pf{t-1})} \\
        &=\Pr(\pf{t-1}|\pf{t})\ \Pr(\pf{t-2}, \cdots\!, \pf{1}|\pf{t-1}) \nonumber \\
        &\quad\times\frac{\Pr(\pf{t}|\pf{t-1}, \cdots\!, \pf{1})}{\Pr(\pf{t}|\pf{t-1})}.
    \end{align*}
    The fraction on the last line equals one because, from Eq.~\eqref{eq:single-qubit-teleportation-pauli-frame:a}, the sequence of Pauli frames is Markovian. Substituting this factorization into Eq.~\eqref{eq:single-qubit-teleportation-chain-conditional-average-channel} and identifying $\overline{\mathcal{N}_{t-1}|\pf{t-1}}$ in the resulting expression gives Eq.~\eqref{eq:single-qubit-teleportation-chain-average-conditional-channel}.
\end{proof}

\begin{lemma}
    \label{lem:leading-approximation}
    For all $t > 1$, $\overline{\mathcal{N}_t}$ can be related to $\overline{\mathcal{N}_{t-1}}$ using the following set of equations:
    \begin{align}
        \overline{\mathcal{N}_t} &= (\mathcal{H}^t\mathcal{E}_{I,t}\mathcal{H}^t)\ \overline{\mathcal{N}_{t-1}} + (\mathcal{H}^t\mathcal{E}_{X,t}\mathcal{H}^t)\ \delta\mathcal{N}_{t-1}, \label{eq:single-qubit-teleportation-chain-iterative-average-error-channel} \\
        \delta\mathcal{N}_t &= (\mathcal{H}^t\mathcal{E}_{Z,t}\mathcal{H}^t)\ \overline{\mathcal{N}_{t-1}} + (\mathcal{H}^t\mathcal{E}_{Y,t}\mathcal{H}^t)\ \delta\mathcal{N}_{t-1}. \label{eq:single-qubit-teleportation-chain-iterative-conditional-average-error-channel-bias}
    \end{align}
    where $\delta\mathcal{N}_t$ is the following linear combination of conditional channels,
    \begin{align}
        \delta\mathcal{N}_t \!&= \!\begin{cases}
            \!\frac{1}{2}(\overline{\mathcal{N}_1|\mathcal{H}} - \overline{\mathcal{N}_1|\mathcal{H}\mathcal{Z}}) & t=1\\
            \!\frac{1}{4}(
                \overline{\mathcal{N}_t|\mathcal{I}} - \overline{\mathcal{N}_t|\mathcal{X}} - \overline{\mathcal{N}_t|\mathcal{Y}} + \overline{\mathcal{N}_t|\mathcal{Z}} 
            ) & \text{even }t \\
            \!\frac{1}{4}(
                \overline{\mathcal{N}_t|\mathcal{H}} + \overline{\mathcal{N}_t|\mathcal{H}\mathcal{X}} - \overline{\mathcal{N}_t|\mathcal{H}\mathcal{Y}} - \overline{\mathcal{N}_t| \mathcal{H}\mathcal{Z}}
            ) & \text{odd }t>1
        \end{cases}. \label{eq:delta-N-definition}
    \end{align}
\end{lemma}
Lemma \ref{lem:leading-approximation} shows that the leading order term for $\mathcal{P}_{(t,t-1)}$ is $\mathcal{H}^t\mathcal{E}_{I,t}\mathcal{H}^t$, which equals the Pauli twirl of $\mathcal{E}_t$ or $\mathcal{H}\mathcal{E}_t\mathcal{H}$. Only this term would exist in Eq.~\eqref{eq:single-qubit-teleportation-chain-iterative-average-error-channel} if randomized compiling is performed at each timestep. The second term in Eq.~\eqref{eq:single-qubit-teleportation-chain-iterative-average-error-channel} is a correction due to the correlation between Pauli frames.

\begin{proof}[Proof of lemma~\ref{lem:leading-approximation}]
    We show the proof for Eq.~\eqref{eq:single-qubit-teleportation-chain-iterative-average-error-channel} for $t=2$. The proofs for Eq.~\eqref{eq:single-qubit-teleportation-chain-iterative-average-error-channel} when $t \neq 2$ and the proof for Eq.~\eqref{eq:single-qubit-teleportation-chain-iterative-conditional-average-error-channel-bias} follow similarly.
    
    \vspace{2mm}
    \noindent\underline{Proof of Eq.~\eqref{eq:single-qubit-teleportation-chain-iterative-average-error-channel} for $t=2$.}
    Substituting Eq.~\eqref{eq:single-qubit-teleportation-chain-conditional-average-channel-iterative} into Eq.~\eqref{eq:single-qubit-teleportation-chain-average-conditional-channel} gives
    \begin{align*}
        \overline{\mathcal{N}_2} &= \sum_{\pf{2}, \pf{1}}
         \Pr(\pf{2}, \pf{1})\ (\pf{2}^{-1}\mathcal{E}_2\pf{2})\ \overline{\mathcal{N}_1|\pf{1}},
    \end{align*}
    where we have set $t=2$, simplified the probabilities, and used Eq.~\eqref{eq:single-qubit-teleportation-transformed-error} for $\mathcal{E}_2'$. We can enumerate four equally probable combinations for $\pf{1}$ and $\pf{2}$. First, we can choose $\pf{1}$ from two equally probable values, $\mathcal{H}$ or $\mathcal{H}\mathcal{Z}$ [Eq.~\eqref{eq:single-qubit-teleportation-chain-uniform-marginal-distribution-one}]. Then, based on the value chosen for $\pf{1}$, we can choose $\pf{2}$ in the two ways allowed by the relation $\pf{2} = \mathcal{H}\mathcal{Z}^{m_2}\pf{1}$ [Eq.~\eqref{eq:single-qubit-teleportation-pauli-frame:a}]. Now we sum over these four combinations,
    \begin{align*}
        \overline{\mathcal{N}_2}
        &= \frac{1}{4}\left[(\mathcal{E}_2 + \mathcal{X}\mathcal{E}_2\mathcal{X})
            \overline{\mathcal{N}_1|\mathcal{H}} + (\mathcal{Z}\mathcal{E}_2\mathcal{Z} + \mathcal{Y}\mathcal{E}_2\mathcal{Y})
            \overline{\mathcal{N}_1|\mathcal{H}\mathcal{Z}}\right] \\
        &= (\mathcal{E}_{I,2} + \mathcal{E}_{X,2})
        \ \overline{\mathcal{N}_1|\mathcal{H}} \big/2 + (\mathcal{E}_{I,2} - \mathcal{E}_{X,2})
        \ \overline{\mathcal{N}_1|\mathcal{H}\mathcal{Z}} \big/2 \\
        &= \mathcal{E}_{I,2}(\overline{\mathcal{N}_1|\mathcal{H}} + \overline{\mathcal{N}_1|\mathcal{H}\mathcal{Z}}) \big/2 + \mathcal{E}_{X,2} (\overline{\mathcal{N}_1|\mathcal{H}} - \overline{\mathcal{N}_1|\mathcal{H}\mathcal{Z}}) \big/2
    \end{align*}
    where plugged in Eq.~\eqref{eq:coherence-transformation} and factored out $\mathcal{E}_{I,2}$ and $\mathcal{E}_{X,2}$. The factor for $\mathcal{E}_{I,2}$ can be identified as $\overline{\mathcal{N}_1}$ [Eq~\eqref{eq:single-qubit-teleportation-chain-average-conditional-channel}], and the factor for $\mathcal{E}_{X,2}$ can be identified as $\delta\mathcal{N}_1$ [Eq.~\eqref{eq:delta-N-definition}]. Therefore,
    \begin{align*}
        \overline{\mathcal{N}_2} = \mathcal{E}_{I,2}\ \overline{\mathcal{N}_1} + \mathcal{E}_{X,2}\ \delta\mathcal{N}_1.
    \end{align*}
    
    \vspace{2mm}
    \noindent\underline{Proof of Eq.~\eqref{eq:single-qubit-teleportation-chain-iterative-average-error-channel} for $t \neq 2$.} The proof is similar to the case $t = 2$ except that there are now four equally probable values for $\pf{t-1}$ [Eqs.~\eqref{eq:single-qubit-teleportation-chain-uniform-marginal-distribution-even}~and~\eqref{eq:single-qubit-teleportation-chain-uniform-marginal-distribution-odd}]. Following the same procedure as case $t=1$ yields
    \begin{align*}
        \overline{\mathcal{N}_t} &= \mathcal{E}_{I,t}\ \overline{\mathcal{N}_{t-1}} + \mathcal{E}_{X,t}\ \delta\mathcal{N}_{t-1} & \text{even }t>2, \\
        \overline{\mathcal{N}_t} &= \mathcal{H}\mathcal{E}_{I,t}\mathcal{H}\ \overline{\mathcal{N}_{t-1}} + \mathcal{H}\mathcal{E}_{X,t}\mathcal{H}\ \delta\mathcal{N}_{t-1} & \text{odd }t>1.
    \end{align*}
    Combining the three cases above proves Eq.~\eqref{eq:single-qubit-teleportation-chain-iterative-average-error-channel}. Next we prove Eq.~\eqref{eq:single-qubit-teleportation-chain-iterative-conditional-average-error-channel-bias}.

    \vspace{3mm}
    \noindent\underline{Proof of Eq.~\eqref{eq:single-qubit-teleportation-chain-iterative-conditional-average-error-channel-bias}.} The definition of $\delta\mathcal{N}_t$ [Eq.~\eqref{eq:delta-N-definition}] is the same as the expression for $\overline{\mathcal{N}_t}$ [Eq.~\eqref{eq:single-qubit-teleportation-chain-average-conditional-channel}] except for minus signs. The same procedure used to show Eq.~\eqref{eq:single-qubit-teleportation-chain-iterative-average-error-channel} can be used to show Eq.~\eqref{eq:single-qubit-teleportation-chain-iterative-conditional-average-error-channel-bias}.
\end{proof}

\begin{lemma}
    \label{lem:next-approximation}
    If the coherent components $\mathcal{E}_{X,t}, \mathcal{E}_{Y,t}, \mathcal{E}_{Z,t}$ of the physical errors are $O(\epsilon)$ where $\epsilon$ is a small positive number, then $\overline{\mathcal{N}_t} = \mathcal{P}_{(t, t-1)}\overline{\mathcal{N}_{t-1}} + O(\epsilon^3)$ where
    \begin{align}
        \mathcal{P}_{(t, t-1)} = \begin{cases}
            \mathcal{E}_{I, t} + \mathcal{E}_{X, t}(\mathcal{H}\mathcal{E}_{Z, t-1}\mathcal{H})\frac{1}{\mathcal{H}\mathcal{E}_{I,t-1}\mathcal{H}} & \text{even }t>2 \\
            \mathcal{H}\mathcal{E}_{I, t}\mathcal{H} + (\mathcal{H}\mathcal{E}_{X, t}\mathcal{H})\mathcal{E}_{Z, t-1}\frac{1}{\mathcal{E}_{I,t-1}} & \text{odd }t>1
        \end{cases}. \label{eq:almost-bounds}
    \end{align}
\end{lemma}
Lemma \ref{lem:next-approximation} shows that we can approximately absorbed the correction term in Eq.~\eqref{eq:single-qubit-teleportation-chain-iterative-average-error-channel} into the expression for $\mathcal{P}_{(t,t-1)}$. The correction term is a product of coherent parts of consecutive error channels $\mathcal{E}_t$ and $\mathcal{E}_{t-1}$, adjusted by the incoherent part of $\mathcal{E}_{t-1}$. This reflects the effect of interference of errors within the short correlation time of the Pauli frames in the teleportation chain.

\begin{proof}[Proof of lemma \ref{lem:next-approximation}]

First we show that $\delta\mathcal{N}_t$ is $O(\epsilon)$ for all $t$. For $t=1$, we plug Eq.~\eqref{eq:single-qubit-teleportation-chain-conditional-average-channel} into Eq.~\eqref{eq:delta-N-definition} and get
\begin{align*}
    \delta\mathcal{N}_1 = \frac{1}{2}(\mathcal{H}\mathcal{E}_1\mathcal{H} - \mathcal{Z}\mathcal{H}\mathcal{E}_1\mathcal{H}\mathcal{Z}) = \frac{1}{2}(\mathcal{H}\mathcal{E}_1\mathcal{H} - \mathcal{H}\mathcal{X}\mathcal{E}_1\mathcal{X}\mathcal{H}).
\end{align*}
Decomposing $\mathcal{E}_1$ using Eq.~\eqref{eq:coherence-transformation} and simplifying gives
\begin{align}
    \delta\mathcal{N}_1 &= \mathcal{H}\mathcal{E}_{Y,1}\mathcal{H} + \mathcal{H}\mathcal{E}_{Z,1}\mathcal{H} = O(\epsilon). \label{eq:lem-3-case-t=1}
\end{align}
Applying Eq.~\eqref{eq:single-qubit-teleportation-chain-iterative-conditional-average-error-channel-bias} then shows that $\delta\mathcal{N}_t = O(\epsilon)$ for all $t > 1$. Now we show Eq.~\eqref{eq:almost-bounds}.

\vspace{2mm}
\noindent\underline{Proof of Eq.~\eqref{eq:almost-bounds} for even $t > 2$.}
We first substitute Eq.~\eqref{eq:single-qubit-teleportation-chain-iterative-conditional-average-error-channel-bias} into Eq.~\eqref{eq:single-qubit-teleportation-chain-iterative-average-error-channel}, which gives
\begin{subequations}
\begin{align}
    &\overline{\mathcal{N}_t} = \mathcal{E}_{I, t}\ \overline{\mathcal{N}_{t-1}} + \mathcal{E}_{X, t}\ (\mathcal{H}\mathcal{E}_{Z, t-1}\mathcal{H})\ \overline{\mathcal{N}_{t-2}} \nonumber \\
    &\quad\ + \mathcal{E}_{X, t}\ (\mathcal{H}\mathcal{E}_{Y, t-1}\mathcal{H})\ \delta\mathcal{N}_{t-2}. \\
    &= \mathcal{E}_{I, t}\ \overline{\mathcal{N}_{t-1}} + \mathcal{E}_{X, t}\ (\mathcal{H}\mathcal{E}_{Z, t-1}\mathcal{H})\ \overline{\mathcal{N}_{t-2}} + O(\epsilon^3). \label{eq:truncated-memory-equation-approximate}
\end{align}
\end{subequations}
We can further simplify this approximation by relating $\overline{\mathcal{N}_{t-2}}$ back to $\overline{\mathcal{N}_{t-1}}$. In Eq.~\eqref{eq:single-qubit-teleportation-chain-iterative-average-error-channel}, setting $t$ to $t-1$ gives
\begin{align*}
    \overline{\mathcal{N}_{t-1}} &= (\mathcal{H}\mathcal{E}_{I,t-1}\mathcal{H})\ \overline{\mathcal{N}_{t-2}} + (\mathcal{H}\mathcal{E}_{X,t-1}\mathcal{H})\ \delta\mathcal{N}_{t-2} \\
    &= (\mathcal{H}\mathcal{E}_{I,t-1}\mathcal{H})\ \overline{\mathcal{N}_{t-2}} + O(\epsilon^2).
\end{align*}
We can take the inverse and approximate
\begin{align*}
    \overline{\mathcal{N}_{t-2}} = \frac{1}{\mathcal{H}\mathcal{E}_{I,t-1}\mathcal{H}}\overline{\mathcal{N}_{t-1}} + O(\epsilon^2).
\end{align*}
Note that here we do not require the inverse map to be physically achievable. Substitution into Eq.~\eqref{eq:truncated-memory-equation-approximate} gives Eq.~\eqref{eq:almost-bounds} for even $t > 2$

\vspace{2mm}
\noindent\underline{Proof of Eq.~\eqref{eq:almost-bounds} for odd $t > 1$.} Proof is similar to even $t > 2$.

\end{proof}

Lemma~\ref{lem:next-approximation} is an informal version of theorem~\ref{thm:bound-on-error-growth-per-timestep}. To prove theorem~\ref{thm:bound-on-error-growth-per-timestep}, we go through the proof of lemma~\ref{lem:next-approximation} again and put strict bounds on the error terms.

\begin{proof}[Proof of theorem~\ref{thm:bound-on-error-growth-per-timestep}]
Because the theorem is stated in terms of the PTM matrix elements, we first express Eqs.~\eqref{eq:single-qubit-teleportation-chain-iterative-average-error-channel}~and~\eqref{eq:single-qubit-teleportation-chain-iterative-conditional-average-error-channel-bias} in terms of the elements of the PTM [Eq.~\eqref{eq:PTM-definition}],
\begin{align}
    \left[\overline{\mathcal{N}_t}\right]_{P',P} &= \sum_{P''}\left[\mathcal{E}_{I,t}^H\right]_{P',P''}
    \left[\overline{\mathcal{N}_{t-1}}\right]_{P'',P} \nonumber \\
    &\quad + \sum_{P''}\left[\mathcal{E}_{X,t}^H\right]_{P',P''}
    \left[\delta\mathcal{N}_{t-1}\right]_{P'',P}, \label{eq:PTM-element-equation-N} \\
    \left[\delta\mathcal{N}_t\right]_{P',P} &= \sum_{P''}\left[\mathcal{E}_{Z,t}^H\right]_{P',P''}
    \left[\overline{\mathcal{N}_{t-1}}\right]_{P'',P} \nonumber \\
    &\quad + \sum_{P''}\left[\mathcal{E}_{Y,t}^H\right]_{P',P''}
    \left[\delta\mathcal{N}_{t-1}\right]_{P'',P}, \label{eq:PTM-element-equation-delta-N}
\end{align}
where we have absorbed the Hadamards into the shorthand notation $\mathcal{E}_{P,t}^H = \mathcal{H}\mathcal{E}_{P,t}\mathcal{H}$. To determine the diagonal elements of the PTM of $\overline{\mathcal{N}_t}$ [Eq.~\eqref{eq:average-infidelity-definition}], we set $P' = P$ in $\left[\overline{\mathcal{N}_t}\right]_{P',P}$. In addition, in the sum over $P''$, only one $P''$ survives, because for even (or odd) $t$, $[\mathcal{E}_{P,t}^H]_{P',P''}$ is non-zero only if $P', P''$ differ by $P$ (or $HPH$) up to a phase. This allows us to arrive at the following matrix equations: for even $t$,
\begin{align}
    \begin{bmatrix}
    \left[\overline{\mathcal{N}_t}\right]_{P,P} \\
    \left[\delta\mathcal{N}_t\right]_{ZP,P}
    \end{bmatrix} \!\!&=\!\! \begin{bmatrix}
    \left[\mathcal{E}_{I,t}^H\right]_{P,P} & \left[\mathcal{E}_{X,t}^H\right]_{P,XP} \\
    \left[\mathcal{E}_{Z,t}^H\right]_{ZP,P} & \left[\mathcal{E}_{Y,t}^H\right]_{ZP,XP}
    \end{bmatrix} \!\! \begin{bmatrix}
    \left[\overline{\mathcal{N}_{t-1}}\right]_{P,P} \\
    \left[\delta\mathcal{N}_{t-1}\right]_{XP,P}
    \end{bmatrix}\!, \label{eq:matrix-equation-even-t}
\end{align}
and for odd $t$
\begin{align}
    \begin{bmatrix}
        \left[\overline{\mathcal{N}_t}\right]_{P,P} \\
        \left[\delta\mathcal{N}_t\right]_{XP,P}
    \end{bmatrix} \!\!&=\!\! \begin{bmatrix}
        \left[\mathcal{E}_{I,t}^H\right]_{P,P} & \left[\mathcal{E}_{X,t}^H\right]_{P,ZP} \\
        \left[\mathcal{E}_{Z,t}^H\right]_{XP,P} & \left[\mathcal{E}_{Y,t}^H\right]_{XP,ZP}
    \end{bmatrix} \!\! \begin{bmatrix}
        \left[\overline{\mathcal{N}_{t-1}}\right]_{P,P} \\
        \left[\delta\mathcal{N}_{t-1}\right]_{ZP,P}
    \end{bmatrix}\!. \label{eq:matrix-equation-odd-t}
\end{align}

Next, we return to the smallness condition on coherence established in Eq.~\eqref{eq:theorem-1-smallness-condition} of theorem~\ref{thm:bound-on-error-growth-per-timestep}, which can be used to restrict the size of the off-diagonal PTM elements appearing in Eqs.~\eqref{eq:matrix-equation-even-t}~and~\eqref{eq:matrix-equation-odd-t}. For example, when $t$ is odd, we have
\begin{align*}
    \left[\mathcal{E}_{X,t}^H\right]_{P,ZP} &= \left[\mathcal{H}\mathcal{E}_{X,t}\mathcal{H}\right]_{P,ZP} = \left[\mathcal{E}_{X,t}\right]_{HPH,HZPH} \\
    &= \left[\mathcal{E}_{X,t}\right]_{HPH,XHPH} =  \left[\mathcal{E}_t\right]_{HPH,XHPH},
\end{align*}
where the first equality is due to the definition of $\mathcal{E}_{X,t}^H$ and the second equality is due to the definition of the PTM elements. Using a similar procedure we can show that $\left[\mathcal{E}_{I,t}^H\right]_{P,P} = \left[\mathcal{E}_t\right]_{HPH,HPH}$. The smallness condition [Eq.~\eqref{eq:theorem-1-smallness-condition}] in theorem~\ref{thm:bound-on-error-growth-per-timestep} guarantees that $\left[\mathcal{E}_t\right]_{HPH,XHPH} \leq \epsilon \left[\mathcal{E}_t\right]_{HPH,HPH}$. Therefore we have
\begin{align*}
    \left[\mathcal{E}_{X,t}^H\right]_{P,XP} \leq \epsilon\left[\mathcal{E}_{I,t}^H\right]_{P,P}.
\end{align*}
Repeating the above for all coherence and all $t$ shows that for all $P$ and even $t$,
\begin{align}
    \left|\left[\mathcal{E}_{X,t}^H\right]_{P,XP}\right|, \left|\left[\mathcal{E}_{Y,t}^H\right]_{ZP,XP}\right|, \left|\left[\mathcal{E}_{Z,t}^H\right]_{ZP,P}\right| < \epsilon\left|\left[\mathcal{E}_{I,t}^H\right]_{P,P}\right|, \label{eq:smallness-condition-even-t}
\end{align}
and for all $P$ and odd $t$
\begin{align}
    \left|\left[\mathcal{E}_{X,t}^H\right]_{P,ZP}\right|, \left|\left[\mathcal{E}_{Y,t}^H\right]_{XP,ZP}\right|, \left|\left[\mathcal{E}_{Z,t}^H\right]_{XP,P}\right| < \epsilon\left|\left[\mathcal{E}_{I,t}^H\right]_{P,P}\right|. \label{eq:smallness-condition-odd-t}
\end{align}

For brevity, we simplify the notation and let the subscript for the PTM be inferred from context:
\begin{align}
    \begin{bmatrix}
    \left[\overline{\mathcal{N}_t}\right] \\
    \left[\delta\mathcal{N}_t\right]
    \end{bmatrix} &= \begin{bmatrix}
    \left[\mathcal{E}_{I,t}^H\right] & \left[\mathcal{E}_{X,t}^H\right] \\
    \left[\mathcal{E}_{Z,t}^H\right] & \left[\mathcal{E}_{Y,t}^H\right]
    \end{bmatrix} \begin{bmatrix}
    \left[\overline{\mathcal{N}_{t-1}}\right] \\
    \left[\delta\mathcal{N}_{t-1}\right]
    \end{bmatrix}, \label{eq:matrix-equation-all-t}
\end{align}
and the smallness condition of the coherence reads,
\begin{align}
    \left|\left[\mathcal{E}_{X,t}^H\right]\right|, \left|\left[\mathcal{E}_{Y,t}^H\right]\right|, \left|\left[\mathcal{E}_{Z,t}^H\right]\right| < \epsilon\left|\left[\mathcal{E}_{I,t}^H\right]\right|. \label{eq:smallness-condition-all-t}
\end{align}

Now that the preparatory steps have been proven, we sill show that 
\begin{align}
    \left|\left[\delta\mathcal{N}_t\right]\right| \leq 3\epsilon \left|\left[\overline{\mathcal{N}_t}\right]\right| \label{eq:delta-N_t-smallness}
\end{align}
for all $t$ using the following proof by induction, mirroring the first step in the proof of lemma~\ref{lem:next-approximation}. The case $t=1$ can be shown from Eq.~\eqref{eq:lem-3-case-t=1}:
\begin{align*}
    \left|\left[\delta\mathcal{N}_1\right]_{XP,P}\right| &= \left|\left[\mathcal{E}_{Y,1}^H\right]_{XP,P} + \left[\mathcal{E}_{Z,1}^H\right]_{XP,P}\right|.
\end{align*}
The first term vanishes because the only the PTM matrix elements of $\mathcal{E}_{Y,1}^H$ are zero unless the input and output Pauli operators differ by $Y$. In addition, by Eq.~\eqref{eq:smallness-condition-odd-t} we have that $\left|\left[\mathcal{E}_{Z,1}^H\right]_{XP,P}\right| \leq \epsilon \left|\left[\mathcal{E}_{I,1}^H\right]_{P,P}\right|$. Therefore $\left|\left[\delta\mathcal{N}_1\right]_{XP,P}\right| \leq \epsilon \left|\left[\mathcal{N}_1\right]_{P,P}\right| \leq 3\epsilon \left|\left[\mathcal{N}_1\right]_{P,P}\right|$. Now that the inequality holds for at least some $t$, we can use the iterative relation for $\delta\mathcal{N}_t$ in the second component of Eq.~\eqref{eq:matrix-equation-all-t} and get the following bound for the subsequent time step,
\begin{align}
    \left|\left[\delta\mathcal{N}_{t+1}\right]\right|
    &\leq \left|\left[\mathcal{E}_{Z,t+1}^H\right]\left[\overline{\mathcal{N}_t}\right]\right|
    + \left|\left[\mathcal{E}_{Y,t+1}^H\right]\left[\delta\mathcal{N}_{t}\right]\right| \\
    &\leq (\epsilon + 3\epsilon^2) \left|\left[\mathcal{E}_{I,t+1}^H\right]\left[\overline{\mathcal{N}_t}\right]\right|. \label{eq:almost-bound-for-delta-N-t+1}
\end{align}
To continue we need a bound for $\left|\left[\overline{\mathcal{N}_t}\right]\right|$ in terms of $\left[\overline{\mathcal{N}_{t+1}}\right]$. This can be achieved with the first component of Eq.~\eqref{eq:matrix-equation-all-t}:
\begin{align*}
    \left[\overline{\mathcal{N}_{t+1}}\right]
    &= \left[\mathcal{E}_{I,t+1}^H\right] \left[\overline{\mathcal{N}_t}\right] + \left[\mathcal{E}_{X,t+1}^H\right] \left[\delta\mathcal{N}_t\right] \\
    &= (1 + \delta)\left[\mathcal{E}_{I,t+1}^H\right] \left[\overline{\mathcal{N}_t}\right],\ \  \delta \in [-3\epsilon^2, +3\epsilon^2]
\end{align*}
On the second line we replace $\left[\mathcal{E}_{X,t+1}^H\right] \left[\delta\mathcal{N}_t\right]$ with the bounds on its magnitude. Next we express $\left[\overline{\mathcal{N}_t}\right]$ in terms of $\left[\overline{\mathcal{N}_{t+1}}\right]$.
Next we can move the factor to the other side of the inequality and take the absolute value
\begin{align}
    &\left[\mathcal{E}_{I,t+1}^H\right]\left[\overline{\mathcal{N}_t}\right] \nonumber \\
    &= \frac{1}{1 +\delta}\left[\overline{\mathcal{N}_{t+1}}\right],\ \ \delta \in [-3\epsilon^2, +3\epsilon^2] \\
    &= (1 + \delta_2)\left[\overline{\mathcal{N}_{t+1}}\right],\ \ \delta_2 \in \left[-\frac{3\epsilon^2}{1 + 3\epsilon^2}, +\frac{3\epsilon^2}{1 - 3\epsilon^2}\right] \label{eq:bounds-for-EIt+1Nt}
\end{align}
Substituting the expression above into the inequality for $\left|\left[\delta\mathcal{N}_{t+1}\right]\right|$ in Eq.~\eqref{eq:almost-bound-for-delta-N-t+1} gives
\begin{align*}
    \left|\left[\delta\mathcal{N}_{t+1}\right]\right|
    &\leq (\epsilon + 3\epsilon^2)|1 + \delta_2| \left|\left[\overline{\mathcal{N}_{t+1}}\right]\right| \leq 3\epsilon \left|\left[\overline{\mathcal{N}_{t+1}}\right]\right|.
\end{align*}
The last inequality uses the fact that for $0 \leq \epsilon \leq 1/3$,
\begin{align*}
    (\epsilon + 3\epsilon^2)\left(1+\frac{3\epsilon^2}{1 - 3\epsilon^2}\right) \leq 3\epsilon.
\end{align*}
Therefore, we must have that $\left|\left[\delta\mathcal{N}_{t}\right]\right| \leq 3\epsilon \left|\left[\overline{\mathcal{N}_t}\right]\right|$ for all $t$.

Now we derive a bound for $\left[\overline{\mathcal{N}_t}\right]$ in terms of $\left[\overline{\mathcal{N}_{t-1}}\right]$. Substituting Eq.~\eqref{eq:matrix-equation-all-t} into itself and taking the first component gives
\begin{align}
    \left[\overline{\mathcal{N}_t}\right] &= \left[\mathcal{E}_{I,t}^H\right]\left[\overline{\mathcal{N}_{t-1}}\right] + \left[\mathcal{E}_{X,t}^H\right] \left[\delta\mathcal{N}_{t-1}\right]  \\
    &= \left[\mathcal{E}_{I,t}^H\right]\left[\overline{\mathcal{N}_{t-1}}\right] \nonumber \\
    &\quad+ \left[\mathcal{E}_{X,t}^H\right] \left(\left[\mathcal{E}_{Z,t-1}^H\right]\left[\overline{\mathcal{N}_{t-2}}\right] + \left[\mathcal{E}_{Y,t-1}^H\right]\left[\delta\mathcal{N}_{t-2}\right]\right) \\
    &= \left[\mathcal{E}_{I,t}^H\right]\left[\overline{\mathcal{N}_{t-1}}\right] + \left[\mathcal{E}_{X,t}^H\right] \left[\mathcal{E}_{Z,t-1}^H\right]\left[\overline{\mathcal{N}_{t-2}}\right] \nonumber \\
    &\quad + \left[\mathcal{E}_{X,t}^H\right]\left[\mathcal{E}_{Y,t-1}^H\right]\left[\delta\mathcal{N}_{t-2}\right]. \label{eq:almost-bounds-three-terms}
\end{align}
We apply the previously derived bounds to the last two terms. For $\left[\mathcal{E}_{X,t}^H\right] \left[\mathcal{E}_{Z,t-1}^H\right]\left[\overline{\mathcal{N}_{t-2}}\right]$, we can use the bounds in Eqs.~\eqref{eq:bounds-for-EIt+1Nt}, replace $t$ by $t - 2$, and multiply by $\left[\mathcal{E}_{X,t}^H\right] \left[\mathcal{E}_{Z,t-1}^H\right]\big/\left[\mathcal{E}_{I,t-1}^H\right]$, which gets us
\begin{align*}
    \left[\mathcal{E}_{I,t-1}^H\right]\left[\overline{\mathcal{N}_{t-2}}\right] &= (1 + \delta_2)\left[\overline{\mathcal{N}_{t-1}}\right], \\
    \left[\mathcal{E}_{X,t}^H\right] \left[\mathcal{E}_{Z,t-1}^H\right]\left[\overline{\mathcal{N}_{t-2}}\right] &= (1 + \delta_2) \frac{\left[\mathcal{E}_{X,t}^H\right] \left[\mathcal{E}_{Z,t-1}^H\right]} { \left[\mathcal{E}_{I,t-1}^H\right]} \left[\overline{\mathcal{N}_{t-1}}\right],
\end{align*}
for some $\delta_2 \in [-\frac{3\epsilon^2}{1 + 3\epsilon^2}, +\frac{3\epsilon^2}{1 - 3\epsilon^2}]$. Next we derive a bound for the third term of Eq.~\eqref{eq:almost-bounds-three-terms}, $\left[\mathcal{E}_{X,t}^H\right]\left[\mathcal{E}_{Y,t-1}^H\right]\left[\delta\mathcal{N}_{t-2}\right]$. We first take the absolute value, then substitute the bounds for $\left|\left[\mathcal{E}_{X,t}^H\right]\right|,\left|\left[\mathcal{E}_{Y,t-1}^H\right]\right|,\left|\left[\delta\mathcal{N}_{t-2}\right]\right|$
\begin{align*}
    \left|\left[\mathcal{E}_{X,t}^H\right]\left[\mathcal{E}_{Y,t-1}^H\right]\left[\delta\mathcal{N}_{t-2}\right]\right| &\leq 3\epsilon^3\left|\left[\mathcal{E}_{I,t}^H\right]\left[\mathcal{E}_{I,t-1}^H\right]\left[\overline{\mathcal{N}_{t-2}}\right]\right| \\
    &\leq 3\epsilon^3(1 + \delta_2)\left|\left[\mathcal{E}_{I,t}^H\right]\left[\overline{\mathcal{N}_{t-1}}\right]\right| \\
    &\leq \frac{3\epsilon^3}{1 - 3\epsilon^2}\left|\left[\mathcal{E}_{I,t}^H\right]\left[\overline{\mathcal{N}_{t-1}}\right]\right|.
\end{align*}
The second set of inequalities is derived by substituting the bounds from Eq.~\eqref{eq:bounds-for-EIt+1Nt} and replacing $t$ by $t - 2$. Therefore,
\begin{align*}
    \left[\mathcal{E}_{X,t}^H\right]\left[\mathcal{E}_{Y,t-1}^H\right]\left[\delta\mathcal{N}_{t-2}\right] &= \delta_1 \left[\mathcal{E}_{I,t}^H\right]\left[\overline{\mathcal{N}_{t-1}}\right]
\end{align*}
where $\delta_1 \in \left[-\frac{3\epsilon^3}{1 - 3\epsilon^2}, +\frac{3\epsilon^3}{1 - 3\epsilon^2}\right]$. Combining the bounds we just derived for the second and third terms of Eq.~\eqref{eq:almost-bounds-three-terms}, we obtain the following expression,
\begin{align*}
    \left[\overline{\mathcal{N}_t}\right] &= \left[\mathcal{E}_{I,t}^H\right]\left[\overline{\mathcal{N}_{t-1}}\right] + (1 + \delta_2) \frac{\left[\mathcal{E}_{X,t}^H\right] \left[\mathcal{E}_{Z,t-1}^H\right]} { \left[\mathcal{E}_{I,t-1}^H\right]} \left[\overline{\mathcal{N}_{t-1}}\right] \nonumber \\
    &\quad + \delta_1\left[\mathcal{E}_{I,t}^H\right]\left[\overline{\mathcal{N}_{t-1}}\right],
\end{align*}
for some $\delta_1 \in \left[-\frac{3\epsilon^3}{1 - 3\epsilon^2}, +\frac{3\epsilon^3}{1 - 3\epsilon^2}\right]$ and $\delta_2 \in [-\frac{3\epsilon^2}{1 + 3\epsilon^2}, +\frac{3\epsilon^2}{1 - 3\epsilon^2}]$.

Finally, let's reconstitute the subscripts by referring to Eq.~\eqref{eq:matrix-equation-even-t} for even timesteps and Eq.~\eqref{eq:matrix-equation-odd-t} for odd timesteps:
\begin{align*}
    \left[\mathcal{E}_{I,t}^H\right] &\to \left[\mathcal{E}_{I,t}^H\right]_{P,P} = \left[\mathcal{E}_t^H\right]_{P,P} \\
    \left[\mathcal{E}_{I,t-1}^H\right] &\to \left[\mathcal{E}_{I,t-1}^H\right]_{P,P} = \left[\mathcal{E}_{t-1}^H\right]_{P,P} \\
    \left[\mathcal{E}_{X,t}^H\right] &\to \begin{cases}
        \left[\mathcal{E}_{X,t}^H\right]_{P,XP} = \left[\mathcal{E}_t^H\right]_{P,XP} & \text{even }t \\
        \left[\mathcal{E}_{X,t}^H\right]_{P,ZP} = \left[\mathcal{E}_t^H\right]_{P,ZP} & \text{odd }t\\
    \end{cases} \\
    \left[\mathcal{E}_{Z,t-1}^H\right] &\to \begin{cases}
        \left[\mathcal{E}_{Z,t-1}^H\right]_{XP,P} = \left[\mathcal{E}_{t-1}^H\right]_{XP,P} & \text{even }t \\
        \left[\mathcal{E}_{Z,t-1}^H\right]_{ZP,P} = \left[\mathcal{E}_{t-1}^H\right]_{ZP,P} & \text{odd }t \\
    \end{cases} \\
    \left[\overline{\mathcal{N}_{t-1}}\right] &\to \left[\overline{\mathcal{N}_{t-1}}\right]_{P,P}
\end{align*}
where recall from the statement of theorem~\ref{thm:bound-on-error-growth-per-timestep} that $\mathcal{E}_t^H = \mathcal{H}^t\mathcal{E}_t\mathcal{H}^t$. Doing this gives the bounds on $[\mathcal{P}_{(t, t-1)}]_{P,P}$ stated in theorem~\ref{thm:bound-on-error-growth-per-timestep}.

Finally, we consider the size of the coherence in $\overline{\mathcal{N}_t}$ by examining Eq.~\eqref{eq:PTM-element-equation-N}. From the simple proof in section~\ref{sec:simple-proof}, we know that the coherence in $\overline{\mathcal{N}_t}$ originates from the fact that the first Pauli-frame $\pf{1}$ can only be $\mathcal{H}$ and $\mathcal{HZ}$. The change in the average error channel $\overline{\mathcal{N}_t}$ introduced by this restriction at the beginning of the teleportation chain has the same nature as $\delta \mathcal{N}_t$, the change in the average error channel caused by conditioning on the most recent Pauli frame. We have already shown $\delta \mathcal{N}_t$ to be $O(\epsilon)$ for all $t$ in lemma~\ref{lem:next-approximation}. Therefore, we conclude that the coherence in $\overline{\mathcal{N}_t}$ should be $O(\epsilon)$.

\end{proof}

\section{Proof of third-order bounds}
\label{sec:tighter-bounds-proof}
In section~\ref{app:bounds-proof} we provided a proof of theorem~\ref{thm:bound-on-error-growth-per-timestep}, which gives a second-order bounds on the diagonal elements of the PTM of the noise channel of the single-qubit teleportation chain. In this section we prove a similar third-order version of the bounds, which relate $\left[\overline{\mathcal{N}_t}\right]_{P,P}$ to $\left[\overline{\mathcal{N}_{t-2}}\right]_{P,P}$. First we substitute Eq.~\eqref{eq:matrix-equation-all-t} into itself and get
\begin{align}
    \begin{bmatrix}
        \left[\overline{\mathcal{N}_t}\right] \\
        \left[\delta\mathcal{N}_{t}\right]
    \end{bmatrix} &= \begin{bmatrix}
        \left[\mathcal{E}_{I,t}^H\right] & \left[\mathcal{E}_{X,t}^H\right] \\
        \left[\mathcal{E}_{Z,t}^H\right] & \left[\mathcal{E}_{Y,t}^H\right]
    \end{bmatrix} \nonumber \\
    &\quad\quad \times
    \begin{bmatrix}
        \left[\mathcal{E}_{I,t-1}^H\right] & \left[\mathcal{E}_{X,t-1}^H\right] \\
        \left[\mathcal{E}_{Z,t-1}^H\right] & \left[\mathcal{E}_{Y,t-1}^H\right]
    \end{bmatrix}
    \begin{bmatrix}
        \left[\overline{\mathcal{N}_{t-2}}\right] \\
        \left[\delta\mathcal{N}_{t-2}\right]
    \end{bmatrix} \\
    &= \begin{bmatrix}
        \tilde{r}_t & \tilde{\gamma}_t \\
        \tilde{\delta}_t & \tilde{a}_t
    \end{bmatrix}
    \begin{bmatrix}
        \left[\overline{\mathcal{N}_{t-2}}\right] \\
        \left[\delta\mathcal{N}_{t-2}\right]
    \end{bmatrix}, \label{eq:single-qubit-iterative-two-time-steps}
\end{align}
where we denote
\begin{align}
    \tilde{r}_t &= \left[\mathcal{E}_{I,t}^H\right]\left[\mathcal{E}_{I,t-1}^H\right] + \left[\mathcal{E}_{X,t}^H\right]\left[\mathcal{E}_{Z,t-1}^H\right], \\
    \tilde{\gamma}_t &= \left[\mathcal{E}_{I,t}^H\right]\left[\mathcal{E}_{X,t-1}^H\right] + \left[\mathcal{E}_{X,t}^H\right]\left[\mathcal{E}_{Y,t-1}^H\right], \\
    \tilde{\delta}_t &= \left[\mathcal{E}_{Z,t}^H\right]\left[\mathcal{E}_{I,t-1}^H\right] + \left[\mathcal{E}_{Y,t}^H\right]\left[\mathcal{E}_{Z,t-1}^H\right], \\
    \tilde{a}_t &= \left[\mathcal{E}_{Z,t}^H\right]\left[\mathcal{E}_{X,t-1}^H\right] + \left[\mathcal{E}_{Y,t}^H\right]\left[\mathcal{E}_{Y,t-1}^H\right].
\end{align}
Then we apply the smallness condition from Eq.~\eqref{eq:smallness-condition-all-t} to derive the following bounds on the sizes of the variables above,
\begin{align}
    \tilde{r}_t &= (1+\eta_1) \left[\mathcal{E}_{I,t}^H\right]\left[\mathcal{E}_{I,t-1}^H\right], \quad |\eta_1| \leq \epsilon^2, \\
    \left|\tilde{\delta}_t\right| &\leq (\epsilon + \epsilon^2) \left|\left[\mathcal{E}_{I,t}^H\right]\left[\mathcal{E}_{I,t-1}^H\right]\right|, \\
    \left|\tilde{\gamma}_t\right| &\leq (\epsilon + \epsilon^2) \left|\left[\mathcal{E}_{I,t}^H\right]\left[\mathcal{E}_{I,t-1}^H\right]\right|, \\
    \left|\tilde{a}_t\right| &\leq 2\epsilon^2 \left|\left[\mathcal{E}_{I,t}^H\right]\left[\mathcal{E}_{I,t-1}^H\right]\right|.
\end{align}

We can use Eqs.~\eqref{eq:single-qubit-iterative-two-time-steps}~and~\eqref{eq:delta-N_t-smallness} to bound $\left[\overline{\mathcal{N}_{t-4}}\right]$ in terms of $\left[\overline{\mathcal{N}_{t-2}}\right]$ as follows,
\begin{align}
    \left[\overline{\mathcal{N}_{t-2}}\right] &= \tilde{r}_{t-2}\left[\overline{\mathcal{N}_{t-4}}\right] + \tilde{\gamma}_{t-2}\left[\delta\mathcal{N}_{t-4}\right] \\
    &= (1 + \eta_1) \left[\mathcal{E}_{I,t-2}^H\right]\left[\mathcal{E}_{I,t-3}^H\right] \left[\overline{\mathcal{N}_{t-4}}\right] \nonumber \\
    &\quad\ + \eta_2 \left[\mathcal{E}_{I,t-2}^H\right]\left[\mathcal{E}_{I,t-3}^H\right]\left[\overline{\mathcal{N}_{t-4}}\right] \\
    &= (1 + \eta_3) \left[\mathcal{E}_{I,t-2}^H\right]\left[\mathcal{E}_{I,t-3}^H\right] \left[\overline{\mathcal{N}_{t-4}}\right] \\
    \left[\overline{\mathcal{N}_{t-4}}\right] &= \frac{1}{1 + \eta_3} \frac{1}{\left[\mathcal{E}_{I,t-2}^H\right]\left[\mathcal{E}_{I,t-3}^H\right]}\left[\overline{\mathcal{N}_{t-2}}\right],
\end{align}
with $|\eta_2| \leq 3\epsilon^2 + 3\epsilon^3$ and $|\eta_3| \leq 4\epsilon^2 + 3\epsilon^3 \leq 5\epsilon^2$.

Now we turn to deriving a bound for $\left[\overline{\mathcal{N}_t}\right]$ in terms of $\left[\overline{\mathcal{N}_{t-2}}\right]$. Substituting the iterative relations in Eq.~\eqref{eq:single-qubit-iterative-two-time-steps} into itself and taking the first component, we have that
\begin{align}
    \left[\overline{\mathcal{N}_t}\right] &= \tilde{r}_t\left[\overline{\mathcal{N}_{t-2}}\right] + \tilde{\gamma}_t\tilde{\delta}_{t-2}\left[\overline{\mathcal{N}_{t-4}}\right] \nonumber \\
    &\quad\ + \tilde{\gamma}_t\tilde{a}_{t-2}\left[\delta\mathcal{N}_{t-4}\right] \\
    &= \tilde{r}_t\left[\overline{\mathcal{N}_{t-2}}\right] + \tilde{\gamma}_t\tilde{\delta}_{t-2}\left[\overline{\mathcal{N}_{t-4}}\right] \nonumber \\
    &\quad\ + \eta_4 \left[\mathcal{E}_{I,t}^H\right]\left[\mathcal{E}_{I,t-2}^H\right]\left[\overline{\mathcal{N}_{t-4}}\right] \\
    &= \tilde{r}_t\left[\overline{\mathcal{N}_{t-2}}\right] + \frac{1}{1 +\eta_3} \frac{\tilde{\gamma}_t\tilde{\delta}_{t-2}}{\left[\mathcal{E}_{I,t-2}^H\right]\left[\mathcal{E}_{I,t-3}^H\right]}\left[\overline{\mathcal{N}_{t-2}}\right] \nonumber \\
    &\quad\ + \frac{\eta_4}{1 + \eta_3} \left[\mathcal{E}_{I,t}^H\right]\left[\mathcal{E}_{I,t-2}^H\right]\left[\overline{\mathcal{N}_{t-2}}\right] \\
    &= \tilde{r}_t\left[\overline{\mathcal{N}_{t-2}}\right] + \frac{1}{1 +\eta_3} \frac{\tilde{\gamma}_t\tilde{\delta}_{t-2}}{\left[\mathcal{E}_{I,t-2}^H\right]\left[\mathcal{E}_{I,t-3}^H\right]}\left[\overline{\mathcal{N}_{t-2}}\right] \nonumber \\
    &\quad\ + \eta_5 \left[\mathcal{E}_{I,t}^H\right]\left[\mathcal{E}_{I,t-2}^H\right]\left[\overline{\mathcal{N}_{t-2}}\right],
\end{align}
with $|\eta_4| \leq (\epsilon + \epsilon^2)(2\epsilon^2)(3\epsilon) = 6\epsilon^4 + 6\epsilon^5$ and $|\eta_5| \leq 18\epsilon^4$.

Therefore the lower and upper bounds for $\left[\overline{\mathcal{N}_t}\right]$ in terms of $\left[\overline{\mathcal{N}_{t-2}}\right]$ when $0 < \epsilon < 1/3$ and $t \geq 4$ are as follows (here we reconstitute the subscripts according to whether $t$ is even or odd),
\begin{widetext}
    \begin{align}
        \left[\overline{\mathcal{N}_t}\right]_{P,P} = \left[\mathcal{P}_{(t,t-2)}\right]_{P,P}\left[\overline{\mathcal{N}_{t-2}}\right]_{P,P}
    \end{align}
    where
    \begin{align}
        &\left[\mathcal{P}_{(t,t-2)}\right]_{P,P}^{(\text{ll})} = \tilde{r}_t + \frac{1}{1+\eta_3}\frac{\tilde{\gamma}_t\tilde{\delta}_{t-2}}{\left[\mathcal{E}_{t-2}^H\right]_{P,P}\left[\mathcal{E}_{t-3}^H\right]_{P,P}} +\eta_5\left[\mathcal{E}_t^H\right]_{P,P}\left[\mathcal{E}_{t-2}^H\right]_{P,P},
    \end{align}
    where $|\eta_3| \leq 5\epsilon^2$ and $|\eta_5| \leq 18\epsilon^4$.
    For even $t$
    \begin{align}
        \tilde{r}_t &= \left[\mathcal{E}_t^H\right]_{P,P}\left[\mathcal{E}_{t-1}^H\right]_{P,P} + \left[\mathcal{E}_t^H\right]_{P,XP}\left[\mathcal{E}_{t-1}^H\right]_{XP,P} \\
        \tilde{\gamma}_t &= \left[\mathcal{E}_t^H\right]_{P,P}\left[\mathcal{E}_{t-1}^H\right]_{P,ZP} + \left[\mathcal{E}_t^H\right]_{P,XP}\left[\mathcal{E}_{t-1}^H\right]_{XP,ZP} \\
        \tilde{\delta}_t &= \left[\mathcal{E}_t^H\right]_{ZP,P}\left[\mathcal{E}_{t-1}^H\right]_{P,P} + \left[\mathcal{E}_t^H\right]_{ZP,XP}\left[\mathcal{E}_{t-1}^H\right]_{XP,P}
    \end{align}
    and for odd $t$ we exchange $XP$ and $ZP$ in the subscripts.
\end{widetext}

\section{Cluster state stabilizers}
\label{app:cluster-state-stabilizers}
The cluster state stabilizers as described in Section~\ref{sec:preliminaries} are $\{\vec{s}_t+\vec{s}_{t-2}+\vec{u}_t(m_{(i,t)})\}$. This is in terms of the vector function $\vec{u}_{t}: \{0,1\}^{n} \rightarrow \{0,1\}^{N_{S_X}}$ for odd $t$ and $\vec{u}_{t}: \{0,1\}^{n} \rightarrow \{0,1\}^{N_{S_Z}}$ for even $t$ which we explicitly write here,
\begin{align}
    (\vec{u}_{t}(m_{(i,t)}))_{p}=\bigoplus _{j\in{S_{X,p}}} m_{j,t},\ p\in[0,N_{S_X}], t\ \mathrm{odd}\\
    (\vec{u}_{t}(m_{i,t}))_{q}=\bigoplus _{j\in{S_{Z,q}}} m_{j,t},\ q\in[0,N_{S_Z}],t\ \mathrm{even}
\end{align}
where $\bigoplus$ denotes addition modulo 2. The sum is over $j$, indices of qubits in the support of the stabilizer $S_{X,p}$ or $S_{Z,q}$. As an example, vector function $\vec{u}_x$ checks the modulo 2 weight of the measurement outcomes of all code qubits $i$ at $t$ in the support of the stabilizer $S_{X,p}$ (for $t$ odd) or $S_{Z,q}$ (for $t$ even).

\section{Circuit-level channels \texorpdfstring{$\mathcal{N}_{z,(\gamma,t,w)}$}{N-z,gamma,t,w} with incoherent errors}
\label{app:thm2-incoherent-errors}

The proof presented in Section~\ref{sec:circuitlvl-eithetaz-errors} was shown for a physical error channel $\mathcal{N}_{z,(\gamma,t,w)}=N_{z,(\gamma,t,w)}\cdot N_{z,(\gamma,t,w)}$. Here, we outline the adaptation of the proof for error channels that
have Kraus rank $>1$ but still follows the conditions of pure Z-coherence in theorem~\ref{thm:2}. An example of such a channel includes a channel with Pauli $Z$ errors and unitary errors $e^{i\theta Z}$. 

\begin{proof}
Let us first take the most general single-qubit channel that satisfies the conditions for pure $Z$-coherence,
\begin{align}
    \label{eq:gen-pure-z-coherent-error}
    \mathcal{N}_{z(\gamma,t,w)}(\cdot) = \sum_{i} c_{i(\gamma,t,w)} N_{i,z(\gamma,t,w)} \cdot N_{i,z(\gamma,t,w)}\\
    N_{i,z(\gamma,t,w)} = \alpha_{i(\gamma,t,w)}I+\beta_{i(\gamma,t,w)}Z
\end{align}
Eq.~\eqref{eq:gen-pure-z-coherent-error} is best seen by examining the equivalent restrictions on the $\chi$-matrix representation~\cite{greenbaum2015introductionquantumgateset}.
As before, $\mathcal{N}_{z,(\gamma,t,w)}$ commutes past $\mathrm{CZ}$ gates in the foliation circuit, and therefore we can still push these errors to their previously defined canonical space-time locations, which are immediately after each teleportation on code qubits and before measurement on ancilla qubits. This composes $W_\gamma$ single-qubit channels together into $\mathcal{N}_{z,(\gamma,t)}=\prod_w^{W_\gamma}\mathcal{N}_{z,(\gamma,t,w)}$. The form of the composed error $\mathcal{N}_{z,(\gamma,t)}$ is
\begin{widetext}
\begin{align}
    &\sum_{i_{W_\gamma}}c_{i_{W_\gamma}(\gamma,t,W_\gamma)}N_{i_{W_\gamma},z(\gamma,t,W_{\gamma})}...\left(\sum_{i_2}c_{i_2(\gamma,t,2)}N_{i_2,z(\gamma,t,2)}\left(\sum_{i_1}c_{i_1(\gamma,t,1)}N_{i_1,z(\gamma,t,1)}(\cdot) N^\dagger_{i_1,z(\gamma,t,1)}\right)N^\dagger_{i_2,z(\gamma,t,2)}\right)...N^\dagger_{i_{W_\gamma},z(\gamma,t,W_\gamma)}\notag
    \\
    &=\sum_{i_{W_\gamma},...,i_1} c_{i_{W_\gamma},(\gamma,t,W_\gamma)}...c_{i_{1},(\gamma,t,1)}N_{i_{W_\gamma},(\gamma,t,W_\gamma)}...N_{i_{1},(\gamma,t,W_1)}(\cdot)N_{i_{1},(\gamma,t,W_1)}^\dagger...N_{i_{W_\gamma},(\gamma,t,W_\gamma)}^\dagger
    \label{eq:expanded-kraus}
\end{align}
\end{widetext}
Here, $i_1, i_2,...i_{W_\gamma}$ are indices that go over different Kraus operators in the error channel in eq.~\ref{eq:gen-pure-z-coherent-error}.
By linearity, we can take each term in eq.~\ref{eq:expanded-kraus} and remove real coefficients $c_{i_{W_\gamma},(\gamma,t,W_\gamma)}...c_{i_{1},(\gamma,t,1)}$ and repeat the proof in section~\ref{sec:circuitlvl-eithetaz-errors} by setting our new space time error at $(\gamma,t)$ to be $N'_{z,(\gamma,t)}[i_1,...,i_{W_\gamma}]=N_{i_{W_\gamma},(\gamma,t,W_\gamma)}...N_{i_{1},(\gamma,t,W_1)}$ for a specific $\{i_1,...i_{W_\gamma}\}$. The number of terms eq.~\ref{eq:expanded-kraus} is a constant for LDPC codes with constant check weight. Using equations~\ref{eq:code-qubit-pauli-channel} and ~\ref{eq:ancilla-qubit-pauli-error} we can calculate the Pauli error probabilities $\mathcal{F}[\{\mathcal{N}'_{z,(\gamma,t)}\}]$ for each $\gamma$ and $t$. Importantly, the proof in sections~\ref{sec:coherent-errors-on-code-qubits} and~\ref{sec:coherent-errors-ancilla} show that this calculation can be done independently for each $\gamma$ and $t$ because the error channels for different $(\gamma,t)$ do not coherently add and are twirled independently. Therefore, we avoid having to compute Pauli error probabilities for a number of space time error configurations that grows exponentially with $(n+N_{S_Z}+N_{S_X})\times 2L$.
After computing $\mathcal{F}[\{\mathcal{N}'_{z,(\gamma,t)}\}]$ for $(\gamma,t)$, we weigh these Pauli error probabilities by multiplying each probability by the coefficient $c_{i_{W_\gamma},(\gamma,t,W_\gamma)}...c_{i_{1},(\gamma,t,1)}$ and add this to the overall Pauli error distribution for the original physical error channel, $\mathcal{N}_{z,(\gamma,t,w)}$. That is,
\begin{widetext}
    \begin{align}
    \mathcal{F}[\mathcal{N}_{z,(\gamma,t,w)}]=\sum_{i_1,...,i_{W_\gamma}}c_{i_{W_\gamma},(\gamma,t,W_\gamma)}...c_{i_{1},(\gamma,t,1)}\mathcal{F}[\{\mathcal{N}'_{z,(\gamma,t)}[i_1,...,i_{W_\gamma}]\}]    
    \label{eq:weighing-paulichannels}
\end{align}
\end{widetext}
where we have used element-wise addition and scalar multiplication of each element of $\mathcal{F}[\mathcal{N'_{z,(\gamma,t)}}]$. The weighing of Pauli error distributions in eq.~\ref{eq:weighing-paulichannels} is possible because we are constructing a joint probability distribution over two independent distributions, one distribution over Pauli errors given by $\mathcal{F}[\{\mathcal{N}'_{z,(\gamma,t)}\}]$, and one distribution over indices $i_1,...i_{W_\gamma}$ given by the set of coefficients $\{c_{i_{W_\gamma},(\gamma,t,W_\gamma)}...c_{i_{1},(\gamma,t,1)}\}\forall i_1,...i_{W_\gamma}$. Therefore, by using the linearity of the channel $\mathcal{N}_{z,(\gamma,t)}$, we can construct a probability distribution that is a linear combination of Pauli error distributions from each term of $\mathcal{N}_{z,(\gamma,t)}$.
\end{proof}

\end{document}